\documentclass[11pt,reqno]{amsart}

\usepackage{amsmath}
\usepackage{amssymb}

\usepackage{epsfig}
\usepackage{amsmath, amsthm, amsfonts}
\usepackage{graphicx}

\numberwithin{equation}{section}

\newcommand{\C}{{\mathbb C}}

\renewcommand{\Re}{{\operatorname{Re\,}}}
\renewcommand{\Im}{{\operatorname{Im\,}}}

\newcommand{\Tr}{{{\operatorname{Tr}}}}

\newcommand{\al}{\alpha}
\newcommand{\be}{\beta}
\newcommand{\ga}{\gamma}
\newcommand{\Ga}{\Gamma}

\newcommand{\ep}{\varepsilon}
\newcommand{\de}{\delta}
\newcommand{\De}{\Delta}

\newcommand{\sg}{\sigma}

\newcommand{\om}{\omega}

\newcommand{\z}{\zeta}

\newtheorem{theo}{{\sc \bf Theorem}}[section]
\newtheorem{cor}[theo]{{\sc \bf Corollary}}

\newtheorem{prop}[theo]{{\sc \bf Proposition}}

\newtheorem{lemma}{Lemma}[section]
\newtheorem{theorem}[lemma]{Theorem}

\newtheorem{proposition}[lemma]{Proposition}
\newtheorem{corollary}[lemma]{Corollary}
\newtheorem{remark}[lemma]{Remark}

\begin{document}

\title[Topological expansion in the cubic random matrix model]
{Topological expansion in the cubic random matrix model}

\author{Pavel Bleher}
\address{Department of Mathematical Sciences,
Indiana University-Purdue University Indianapolis,
402 N. Blackford St., Indianapolis, IN 46202, U.S.A.}
\email{bleher@math.iupui.edu}

\author{Alfredo Dea\~no}
\address{Departamento de Matem\'{a}ticas,
Universidad Carlos III de Madrid,
Avda. de la Universidad, 30. 28911 Legan\'{e}s, Madrid, Spain}
\email{alfredo.deanho@math.uc3m.es}

\thanks{The first author is supported in part
by the National Science Foundation (NSF) Grant and DMS-0969254. The second author acknowledges financial support from Universidad Carlos III de Madrid (Ayudas para la Movilidad del Programa Propio de Investigaci\'{o}n 2010) and project MTM2009-11686 from the Spanish Ministry of Science and Innovation}

\thanks{Both authors would like to thank the Mathematical Sciences Research Institute for their hospitality during the program ``Random Matrices, Interacting Particle Systems and Integrable Systems" in the Fall of 2010, where much of this work was performed. The authors acknowledge useful discussions with N. M. Ercolani, B. Eynard, A. B. J. Kuijlaars and K. T-R. McLaughlin.}

\date{\today}

\begin{abstract}
In this paper we study the topological expansion in the cubic random matrix model,
and we evaluate explicitly the expansion coefficients for genus 0 and 1. For genus 0
our formula coincides with the one of Br\'ezin, Itzykson, Parisi, and Zuber \cite{BIPZ}.
For the expansion coefficients of higher genus we obtain their asymptotic behavior as the number
of vertices of the associated graphs tends to infinity. Our study is based on the Riemann-Hilbert
problem, string equations, and the Toda equation.
\end{abstract}

\keywords{Random matrices, asymptotic representation in the complex domain, topological expansion, classical hypergeometric functions, Riemann-Hilbert problem, string equations, Toda equation. 2010 MSC: 30E15, 60B20, 33C05, 05C10.}

\maketitle

\section{Introduction and statement of the main results}
In this paper we return to the classical work \cite{BIPZ} by Br\'ezin, Itzykson, Parisi and Zuber, in which,
among other things, the authors explicitly calculated the coefficients
of the topological expansion in the cubic random matrix model in genus 0.
Our main goal will be to rigorously prove the results of \cite{BIPZ} and to obtain an explicit formula for
the coefficients of the topological expansion in genus 1. We will also prove some formulae and asymptotic results
for the coefficients of the topological expansion in higher genera.

We consider the random matrix model given by the probability distribution
\begin{equation}\label{dPM}
d\mu_N(M)=\frac{1}{\tilde Z_N}e^{-N \textrm{Tr} V(M)}dM,
\end{equation}
on the space of $N\times N$ Hermitian matrices $M$, where
\begin{equation}\label{VM}
V(M)=\frac{M^2}{2}-uM^3,
\end{equation}
and $u>0$. The model is ill-defined because of the divergence at infinity of the partition
function,
\begin{equation}\label{tZN}
\tilde Z_N(u)=\int e^{-N \textrm{Tr} V(M)}dM.
\end{equation}
The partition function of eigenvalues,
\begin{equation}\label{ZN0}
Z_N(u)=\int_{-\infty}^\infty\ldots\int_{-\infty}^\infty \prod_{1\leq j<k\leq N}(z_j-z_k)^2\,
\prod_{j=1}^N e^{-N \left(\frac{z_j^2}{2}-uz_j^3\right)}dz_1\ldots dz_N
\end{equation}
diverges as well on the real line. To regularize it, we will consider integration on a specially
chosen contour $\Ga$ in the complex plane:
\begin{equation}\label{ZN}
Z_N(u)=\int_\Ga\ldots\int_\Ga \prod_{1\leq j<k\leq N}(z_j-z_k)^2\,
\prod_{j=1}^N e^{-N \left(\frac{z_j^2}{2}-uz_j^3\right)}dz_1\ldots dz_N,
\end{equation}
on which the integral converges.
\begin{center}
\begin{figure}[h]
\begin{center}
\scalebox{0.52}{\includegraphics{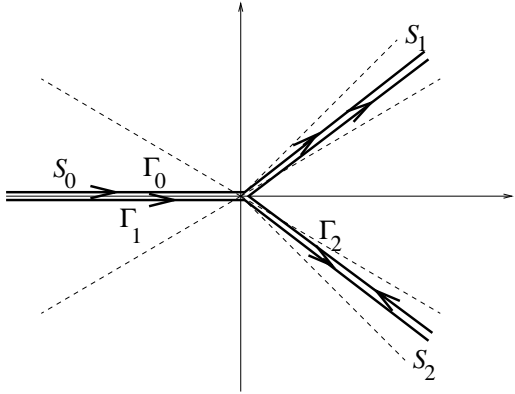}}
\end{center}
  \caption[sectors ]{The sectors $S_0$, $S_1$, $S_2$ and the contours  $\Ga_0$, $\Ga_1$, $\Ga_2$.}
 \end{figure}
\label{Figsectors}
\end{center}

To choose $\Ga$, consider the three sectors on the complex plane,
\begin{equation}\label{sectors}
\begin{aligned}
S_0&=\Big\{z\in\C:\;\frac{5\pi}{6}< \arg z < \frac{7\pi}{6}\Big\}\,,\\
 S_1&=\Big\{z\in\C:\;\frac{\pi}{6}< \arg z < \frac{\pi}{4}\Big\}\,,\\
 S_2&=\Big\{z\in\C:\;-\frac{\pi}{4}< \arg z < -\frac{\pi}{6}\Big\}\,,
\end{aligned}
\end{equation}
see Fig.\ref{Figsectors}.
Then for any ray
\begin{equation}\label{ray}
R_{\theta}=\big\{z\in\C:\; \arg z =\theta \big\},
\end{equation}
lying in the sectors $S_0$, $S_1$, and $S_2$, the integral
\begin{equation}\label{integral}
\int_{R_{\theta}}z^k e^{-N \left(\frac{z^2}{2}-uz^3\right)}dz
\end{equation}
converges for any $k=0,1,\ldots$ and any $u\ge 0$. Therefore, we will use contours consisting of two
rays in the sectors $S_0$, $S_1$, and $S_2$.
More specifically, we may consider the three contours on the complex plane,
\begin{equation}\label{contours}
\Ga_0=R_\pi\cup R_{\pi/5},\qquad
\Ga_1=R_\pi\cup R_{-\pi/5},\qquad
\Ga_2=R_{-\pi/5}\cup R_{\pi/5},
\end{equation}
with orientation from $(-\infty)$ to $(\infty\,e^{\pi i/5})$ on $\Ga_0$,
from $(-\infty)$ to $(\infty\,e^{-\pi i/5})$ on $\Ga_1$, and
from $(\infty e^{-\pi i/5})$ to $(\infty\,e^{\pi i/5})$ on $\Ga_2$,  see Fig. \ref{Figsectors}.
However,
to get the topological expansion of the free energy, we will use $\Ga=\Ga_0$ or $\Ga=\Ga_1$, but not $\Ga=\Ga_2$.

More generally, following Duits and Kuijlaars \cite{DK},
it is convenient to introduce a linear combination of the contours $\Ga_0$, $\Ga_1$. To that end, let us
fix some $\al\in\C$ and consider $\Ga$ as a linear combination of $\Ga_0$, $\Ga_1$,
\begin{equation}\label{Gammal}
\Ga=\al\Ga_0+(1-\al)\Ga_1,
\end{equation}
in the sense that
\begin{equation}\label{Gamma2}
\int_\Ga f(z)dz=\al\int_{\Ga_0}f(z)dz+(1-\al)\int_{\Ga_1}  f(z) dz,
\end{equation}
cf. \cite{BEH}.
With
this choice of $\Ga=\Ga(\al)$, the integral $Z_N(u)=Z_N(u;\al)$ in \eqref{ZN} is convergent for any $u\ge 0$.
By the Cauchy theorem, we have some flexibility in the choice of the contours $\Ga_0$, $\Ga_1$
within the sectors $S_0$, $S_1$, $S_2$.

It is easy to see, by differentiation of the integral with respect to $u$,
that the partition function $Z_N(u;\al)$ is analytic for $u>0$, and it is infinitely differentiable
for $u\ge 0$. Observe, however, that $Z_N(u;\al)$ is not analytic at $u=0$. Indeed, the analytic
continuation of $Z_N(u)$, $u>0$, to $u=r e^{-3i\theta}$, $r>0$,  can be obtained by the change
of variable $z_j=w_je^{i\theta}$:
\begin{equation}\label{ZNr}
\begin{aligned}
Z_N(r e^{-3i\theta};\al)&=\int_\Ga\ldots\int_\Ga \prod_{1\leq j<k\leq N}(w_je^{i\theta}-w_ke^{i\theta})^2\,\\
&\times\prod_{j=1}^N e^{-N \left(\frac{w_j^2e^{2i\theta}}{2}-rw_j^3\right)}d(w_1e^{i\theta})\ldots d(w_Ne^{i\theta}).
\end{aligned}
\end{equation}
By the Cauchy theorem, one can use the same contour $\Ga=\Ga(\al)$ for all $\theta$. For $\theta=\frac{\pi}{2}$,
the exponential term in the latter formula becomes
\[
e^{N \left(\frac{w_j^2}{2}+rw_j^3\right)}
\]
and the integral in \eqref{ZNr} diverges at $r=0$. Moreover,
\begin{equation}\label{ZNr1}
\begin{aligned}
\lim_{r\to 0}|Z_N(r e^{-3i\theta};\al)|=\infty,\qquad \theta=\frac{\pi}{2}\,,
\end{aligned}
\end{equation}
hence $Z_N(u;\al)$ is not analytic at $u=0$.

We define the free energy as
\begin{equation}\label{free}
F_N(u;\al)=\frac{1}{N^2}\ln \frac{Z_N(u;\al)}{Z_N(0;\al)}.
\end{equation}
Observe that $Z_N(0;\al)$ is independent of $\al$.
The aim of this paper is to investigate the large $N$ asymptotic behavior of the free energy, and, in particular, the structure of the different terms that appear in this asymptotic expansion. We consider $u$ in the interval
$0\leq u<u_c$, where $u_c$ is the following critical value:
\begin{equation}\label{uc}
u_c=\frac{3^{1/4}}{18}.
\end{equation}

We prove the following results:
\begin{theorem}\label{Th1} Suppose that $\al\in\C$ is fixed. Then
for any  $0\leq u< u_c$,
the free energy $F_N(u;\al)$ admits an asymptotic expansion in inverse powers of $N^2$:
\begin{equation}\label{top1}
F_N(u;\al)\sim \sum_{g=0}^{\infty}\frac{ F^{(2g)}(u)}{N^{2g}},
\end{equation}
which is uniform in the variable $u$ on any interval $[0,u_c-\varepsilon]$, $\varepsilon>0$
and which can be differentiated on $[0,u_c-\varepsilon]$ with respect to $u$ any number of times,
with a uniform estimate of the error term  with respect to $u\in [0,u_c-\varepsilon]$.
The functions $F^{(2g)}(u)$ do not depend on $\al$ and they
admit an analytic continuation to the disk $|u|<u_c$ in the complex plane,
and if we expand them in powers of $u$,
\begin{equation}\label{top2}
F^{(2g)}(u)=\sum_{j=1}^\infty \frac{f^{(2g)}_{2j}u^{2j}}{(2j)!},
\end{equation}
then the coefficient $f^{(2g)}_{2j}$ is a positive integer number that counts the number
of $3$-valent connected graphs with $2j$ vertices on a Riemann surface of genus $g$.
\end{theorem}

The existence of the $\frac {1}{N^2}$ asymptotic expansion for the free energy is known in the
physical literature since the classical work of  Bessis, Itzykson, Zuber \cite{BIZ}
(see also references therein to the earlier works).
A rigorous proof of the $\frac {1}{N^2}$
asymptotic expansion for the free energy for a general polynomial $V(M)$ of even degree
is given in the paper of Ercolani and McLaughlin \cite{EML}.
Our proof of Theorem \ref{Th1} follows the work \cite{EML}, with some modifications
related to the fact that the cubic model is defined on the contour $\Ga$ in the complex plane and
not on the real line.

It is worth noticing that the coefficients $F^{(2g)}(u)$ of asymptotic expansion \eqref{top1}
in powers of $1/N^2$ do not depend on $\al$ for all $0\le u<u_c$. The dependence on $\al$ arises, however,
in the double scaling limit as $u\to u_c$. We will return to this question in a subsequent paper.
(See also the paper \cite{DK} of Duits and Kuijlaars, where a similar double scaling limit is studied
for the quartic random matrix model.)

In the next theorem we evaluate explicitly the coefficients $f^{(0)}_{2j}$ in the genus $0$ term $F^{(0)}(u)$:

\begin{theorem}\label{Th2}
The coefficient $f^{(0)}_{2j}$ can be written as
\begin{equation}\label{Th2_1}
f^{(0)}_{2j}=\frac{72^j
\Gamma(\frac{3j}{2})(2j)!}{2\Gamma(j+3)\Gamma(\frac{j}{2}+1)},
\end{equation}
and it has the following asymptotic behavior as $j\to\infty$:
\begin{equation}\label{asympf02j}
\begin{aligned}
f^{(0)}_{2j} &=\frac{K_0(2j)!}{u_c^{2j} j^{7/2}}\left(1+\mathcal{O}(j^{-1})\right), \qquad K_0=\frac{1}{\sqrt{6\pi}},
\end{aligned}
\end{equation}
where $u_c$ is the critical value defined in \eqref{uc}.
\end{theorem}

From formula \eqref{Th2_1} we obtain the first several terms of the Taylor series of $F^{(0)}(u)$ at $u=0$:
\begin{equation}\label{seriesF0}
\begin{aligned}
F^{(0)}(u) &=6u^2+216 u^4+13608u^6+1119744u^8+\frac{540416448}{5}u^{10}\\
&+\mathcal{O}(u^{12}).
\end{aligned}
\end{equation}

Our formula \eqref{Th2_1} for $f^{(0)}_{2j}$ coincides with the one obtained in
the work of Br\'ezin, Itzykson, Parisi and Zuber \cite[\S 4]{BIPZ} in the planar diagrams
approximation.

A more complicated expression, but still explicit, is available when $g=1$:

\begin{theorem}\label{Th3}
The coefficient $f^{(2)}_{2j}$ can be written explicitly in terms of a $_3F_2$ hypergeometric function as follows:
\begin{equation}\label{Th3_1}
f^{(2)}_{2j}=\frac{5\cdot 72^j\Gamma(\frac{3j}{2})(2j)!}{48(3j+2)\Gamma(j+1)\Gamma(\frac{j}{2}+1)}
\, _3F_2\left(\begin{array}{l} -j+1,2,6\\ 5,-\frac{3j}{2}+1 \end{array};\frac{3}{2}\right),
\end{equation}
and it has the following asymptotic behavior as $j\to\infty$:
\begin{equation}\label{Th3_2}
\begin{aligned}
f^{(2)}_{2j} &=\frac{K_2 (2j)!}{u_c^{2j}j}\left(1+\mathcal{O}(j^{-1/2})\right),
\qquad
K_2=\frac{1}{48}\,.
\end{aligned}
\end{equation}
\end{theorem}

\begin{remark}
The $_3F_2$ function in formula \eqref{Th3_1} can actually be written as a linear combination of two $_2F_1$ functions:
\begin{equation}
\begin{aligned}
_3F_2\left(\begin{array}{l} 2,6,-j\\ 5,-\frac{3j}{2}-\frac{1}{2} \end{array};\frac{3}{2}\right)
=&\left[ _2F_1\left(\begin{array}{l} -j,2\\ -\frac{3j}{2}-\frac{1}{2} \end{array};\frac{3}{2}\right)\right.\\
&\left.+\frac{6j}{5(3j+1)}\, _2F_1\left(\begin{array}{l} -j+1,3\\ -\frac{3j}{2}+\frac{1}{2} \end{array};\frac{3}{2}\right)\right].
\end{aligned}
\end{equation}
\end{remark}

From formula \eqref{Th3_1}, the first several terms of the Taylor series of $F^{(2)}(u)$ at $u=0$ are
\begin{equation}\label{F2}
\begin{aligned}
F^{(2)}(u)&=\frac{3}{2}u^2+189u^4+26892u^6+4076568u^8+\frac{3213210384}{5}u^{10}\\
&+\mathcal{O}(u^{12}).
\end{aligned}
\end{equation}

In the case of genus $g>1$, an explicit formula for
the coefficients $f^{(2g)}_{2j}$ becomes complicated,
but it is possible to obtain the asymptotic behavior of these coefficients when $j\to\infty$:

\begin{theorem}\label{Thgg1}
For any $g>1$, the coefficient $f^{(2g)}_{2j}$ has the following asymptotic behavior as $j\to\infty$:
\begin{equation}\label{gg1}
\begin{aligned}
f^{(2g)}_{2j} &=\frac{K_{2g} (2j)! j^{\frac{5g-7}{2}}}{u_c^{2j}}\left(1+\mathcal{O}(j^{-1/2})\right),
\end{aligned}
\end{equation}
where
\begin{equation}\label{gg2}
K_{2g}=\frac{6\cdot 3^{1/4}C_{2g}}{\Gamma\left(\frac{5g-1}{2}\right)u_c^{g}},
\end{equation}
and $C_{2g}$ satisfies the recurrence relation
\begin{equation}\label{gg3}
\begin{aligned}
C_{2g}&=\frac{1}{2^{3/2}3^{5/4}}\bigg(\frac{\left(5g-6\right)\left(5g-4\right)C_{2g-2}}{48}
+54\mathop{\sum_{m+m'=g}}_{m,m'\le g-1} C_{2m}C_{2m'}\bigg)
\end{aligned}
\end{equation}
for $g\ge 1$, with the initial value
\begin{equation}\label{gg4}
C_0=-2^{-1/2}3^{-7/4}.
\end{equation}
\end{theorem}

For $g=2$ we have that
\begin{equation}\label{gg5}
F^{(4)}=\frac{8505 u^6}{2}+2217618 u^8 +\frac{3905028468 u^{10}}{5}+\mathcal O(u^{12}),
\end{equation}
and the constants $C_4$ and $K_4$ are
\begin{equation}\label{gg6}
C_4=\frac{49\cdot 2^{1/2}3^{3/4}}{17915904}\,,\qquad K_4=\frac{7}{1440\sqrt{6\pi}}\,.
\end{equation}

We note that if we  write the recursion \eqref{gg3} as follows:
\begin{equation}
C_{2g}=\mu\left(5g-6\right)\left(5g-4\right)C_{2g-2}+
\nu\mathop{\sum_{m+m'=g}}_{m,m'\le g-1} C_{2m}C_{2m'},
\end{equation}
and if we construct the following generating function (cf. \cite{BE})
\begin{equation}\label{genfun}
y(t)=\sum_{g=0}^{\infty} C_{2g}t^{\frac{1-5g}{2}},
\end{equation}
then $y(t)$ satisfies the Painlev\'e I differential equation
\begin{equation}\label{ode}
y''(t)=\frac{y^2(t)}{8\mu}-\frac{C_0^2}{8\mu}t.
\end{equation}
We can make the change of variables $t=-c \tau$ and $u=\lambda y$ to bring it to the standard form
\begin{equation}\label{PI}
u''(\tau)=6u^2(\tau)+\tau,
\end{equation}
see \cite{dlmf}. Explicitly,
\begin{equation}
c=2^{-3/5}, \qquad \lambda=2^{3/10}3^{5/4}.
\end{equation}

The proof of Theorem \ref{Th1} relies on the fact that the equilibrium measure for this problem is supported by a single interval $[a,b]$ on the real axis, provided $0\leq u< u_c$.
This will be shown in Section \ref{eq}. As a consequence, we will be able to apply the technique
developed in \cite{EML} to the proof of Theorem \ref{Th1}.

The precise structure of the coefficients $F^{(2g)}(u)$ of the asymptotic expansion of
the free energy will be analyzed using the following steps:
\begin{itemize}
\item First we apply the Deift--Zhou nonlinear steepest descent
method to the associated Riemann--Hilbert problem, in order to obtain the large $N$ asymptotic expansion of the recurrence coefficients $\gamma^2_n$ and $\beta_n$ of the corresponding orthogonal polynomials $P_n(z)$, when the index $n$ is of the order of $N$. The coefficients in this asymptotic expansion
will be found by using the string equations, which are nonlinear algebraic relations for $\gamma^2_n$ and $\beta_n$. Simultaneously, this proves the existence of the orthogonal polynomials $P_n(z)$ when the index $n$ is of the order of $N$ and $N$ is large enough.
 \item Then we make a change of variable to reduce the cubic polynomial $V(M)=\frac{M^2}{2}-uM^3$ to
$\tilde V(M)=tM-\frac{M^3}{3}$, and we derive simple formulae for the partition function and the recurrence coefficients under this change of variable. The change of variable allows us
to use the Toda equation that connects the second derivative of the free energy $\tilde F_N$ with respect to the parameter $t$
to the recurrence coefficient $\gamma^2_N$.
\item By integrating the large $N$ asymptotic expansion for $\gamma^2_N$ term by term, we obtain detailed information about the large $N$ asymptotic expansion of the free energy $F_N(u)$, and in particular the form of the different $F^{(2g)}(u)$ terms. Note that in this derivation we only need the recurrence coefficients $\gamma^2_N$ and $\beta_N$ for large values of $N$, whose existence and asymptotic behavior are proved via the Riemann--Hilbert analysis before.
\end{itemize}

We note that an alternative method to derive formulae for $F^{(2g)}$ was proposed in \cite{Amb},
using the so called loop equations (see \cite{EMLloop} for a rigorous derivation of the loop equations).
In \cite{Amb}, explicit expressions for $F^{(2)}$ and $F^{(4)}$ are given,
in terms of elementary functions that depend on the endpoints of the support of the equilibrium measure
associated with the corresponding potential. In \cite{CE,Eyn}, a general heuristic formula for computing
$F^{(g)}$ is presented for arbitrary potentials and multicut cases.

We would like also to bring attention to the very recent paper \cite{EP} of Ercolani and Pierce, which
appeared after the present paper had been posted on arXiv, and which cited the
present paper.  In \cite{EP}
a different approach, based on the so called difference string equations, is developed for the calculation of the terms of the topological expansion in the
random matrix model with a general potential $V(M)$.

The rest of the paper is organized as follows: in Sections \ref{eq} and \ref{OPs} we analyze the equilibrium measure and
the corresponding Riemann-Hilbert problem for orthogonal polynomials with respect to the cubic-type weight. In Section \ref{rec_coeff} we
apply the string equations to the large $N$ expansion of the recurrence coefficients $\gamma_n^2$ and $\beta_n$.
This is used in Section \ref{free_energy} to obtain the large $N$ expansion of the free energy $F_N(u)$,
and this leads to the proof of Theorems \ref{Th2} and Theorem \ref{Th3} in Section \ref{Th2Th3}.
In Section \ref{Th4} we prove Theorem \ref{Thgg1}, expanding the recurrence coefficients around
the critical point $u=u_c$ and selecting the terms that give the dominant behavior in the free energy.
Finally, Section \ref{graphs} is devoted to the interpretation of the topological expansion in terms
of connected graphs embedded in closed Riemann surfaces.

\section{Equilibrium measure and the Riemann--Hilbert analysis}\label{eq}

\subsection{Construction}

Let us denote by $\varrho(s)$ the density of the equilibrium measure for this problem, that we assume supported on a certain curve $J\subset\mathbb{C}$ with endpoints $z=a$ and $z=b$. We consider $J$ oriented from $a$ to $b$, and we take the $+$-side on the left of $J$ and the $-$-side on the right of $J$, following the standard convention.

The resolvent
\begin{equation}
\om(z)=\int_J \frac{\varrho(s)ds}{z-s}
\end{equation}
is analytic off $J$, and it satisfies the Euler--Lagrange equation on $J$,
\begin{equation}\label{int1}
\begin{aligned}
&\om_+(z)+\om_-(z)=V'(z),\qquad z\in J,
\end{aligned}
\end{equation}
and the asymptotics at infinity,
\begin{equation}\label{int2}
\om(z)=\frac{1}{z}+\mathcal{O}(z^{-2}),\quad z\to\infty.
\end{equation}
To solve (\ref{int1}), we write
\begin{equation}\label{wRh}
\om(z)=\frac{1}{2}V'(z)-\frac{1}{2}\sqrt{R(z)}\,h(z),
\end{equation}
where $h(z)$ is an analytic function and
\begin{equation}\label{int3}
R(z)=(z-a)(z-b),
\end{equation}
since we are assuming that we are in the one cut case. We take the principal sheet for $\sqrt{R(z)}$, with a cut on $J$.
Due to the Plemelj formula, we have
$$
\om_-(z)-\om_+(z)=2\pi i\varrho(z)=\sqrt{R(z)}h(z)\Rightarrow \varrho(z)=\frac{1}{2\pi i}\sqrt{R(z)}h(z).
$$
We easily deduce from \eqref{wRh} that $h(z)$ is a polynomial of degree $1$.
Moreover, if we write
\begin{equation}\label{int4}
a=x-y, \qquad b=x+y
\end{equation}
and $h(z)=A(z-z_0)$, then identifying the asymptotics at infinity on both sides of \eqref{wRh}, we find
\begin{equation}\label{system}
\begin{aligned}
A &=-3u,\\
z_0 &=\frac{1}{3u}-x,\\
2x^2+y^2&=\frac{2x}{3u}\,,\\
\frac{y^2}{4}\left(1-6ux\right)&=1.
\end{aligned}
\end{equation}
Thus,
\begin{equation}\label{int5}
\begin{aligned}
\om(z)&=\frac{z-3uz^2}{2}-\frac{1}{2}\sqrt{(z-a)(z-b)}\,(1-3uz-3ux)
\end{aligned}
\end{equation}
and therefore
\begin{equation}\label{int6}
\begin{aligned}
\varrho(z)&=\frac{1}{2\pi i}\sqrt{(z-a)(z-b)}\,(1-3uz-3ux).
\end{aligned}
\end{equation}
It follows that $\varrho(z)^2$ has a double zero at
\[
z_0=\frac{1}{3u}-x
\]
 and two simple roots
located at $z_{1,2}=a,\, b$.

Next we will prove the following facts:
\begin{itemize}
 \item For $0\le u<u_c$, where $u_c$ is given by \eqref{uc}, the equilibrium measure is supported
on an interval $J=[a,b]$ of the real axis. The endpoints of this interval are analytic functions
of the parameter $u$.
\item For $0\le u<u_c$, we can extend the interval $J$ to an unbounded contour
$\tilde \Gamma_0=\gamma_1\cup J\cup \gamma_2$ in the upper half-plane,
\[
\mathbb{C}_+=\{z\in\C:\; \Im z\ge 0\},
\]
 in such a way that
\begin{equation}\label{Rephi}
\begin{aligned}
 \phi_1(z)>0 & \qquad z\in\gamma_1,\\
 \phi_2(z)>0 & \qquad z\in\gamma_2,
\end{aligned}
\end{equation}
where
\begin{equation}\label{phi12}
\begin{aligned}
\phi_1(z)&=\frac{1}{2}\int_a^z \sqrt{R(s)}h(s)ds ,\\
\phi_2(z)&=\frac{1}{2}\int_b^z \sqrt{R(s)}h(s)ds.\\
\end{aligned}
\end{equation}
Additionally, as $|z|\to\infty$ on $\gamma_1$ and $\gamma_2$, it is possible to estimate the growth of $\Re \phi_1(z)$ and $\Re \phi_2(z)$:
\begin{equation}\label{growthphi}
\Re \phi_i(z)=-\frac{uz^3}{3}+\mathcal{O}(z^2), \qquad i=1,2.
\end{equation}
This last result is relevant in the steepest descent method applied to the Riemann--Hilbert problem,
in order to show that the jump matrices outside of $J$ tend to the identity when $N\to\infty$.

\end{itemize}
\subsection{Support of the equilibrium measure}
Combining the last two equations
in \eqref{system}, we obtain the following cubic equation for the variable $x$ as a function of the parameter $u$:
\begin{equation}\label{cubic}
18u^2x^3-9ux^2+x-6u=0.
\end{equation}
If $u>0$ is small, then this equation has the three solutions,
\begin{equation}\label{cubic2}
x_1=6u+\mathcal O(u^2),\qquad x_2=\frac{1}{6u}+\mathcal O(1),\qquad x_3=\frac{1}{3u}+\mathcal O(1).
\end{equation}
In what follows we will be interested in the first root, $x=x_1$, which is the solution that remains bounded when $u$ is small. More terms in the expansion of $x$ are
\begin{equation}\label{cubic3x}
x=x_1(u)=6u+324u^3+31104u^5+\mathcal O(u^7).
\end{equation}
From the last equation in \eqref{system} we find
\begin{equation}\label{cubic3y}
y=2+36u^2+2916u^4+\mathcal O(u^6).
\end{equation}
The support of the equilibrium measure is the interval $[a,b]=[x-y,x+y]$. As $u\to 0$,
it converges to $[-2,2]$, which is the support of the equilibrium measure for GUE.

The series that we obtain in \eqref{cubic3x} can be proved to be convergent for small values of $u$:

\begin{proposition}\label{Propxu}
The series \eqref{cubic3x} is convergent for $|u|<u_c$, where
$u_c$ is given by \eqref{uc}.
\end{proposition}

\begin{proof} Since by \eqref{cubic2} the root $x=x_1(u)$ is isolated for small $u$, $x_1(u)$ is
analytic at $u=0$.
The discriminant of cubic equation \eqref{cubic} is $\Delta=9u^2(1-34992u^4)$, and it vanishes
at the critical value $u=u_c$. Hence it does not vanish for $0<|u|<u_c$, and the root $x_1(u)$
remains isolated. This proves that $x_1(u)$ is analytic in the disk $|u|<u_c$,
and hence the Taylor series for $x_1(u)$ is convergent in this disk.
\end{proof}

It follows from the last equation of \eqref{system} that a similar result is valid for the variable $y=y(u)>0$.
Namely, $y(u)$ is obviously analytic at $u=0$. Suppose that $1-6ux_1(u)=0$, then from equation \eqref{cubic}
we obtain that $u=0$, which contradicts $1-6ux_1(u)=0$. Therefore,
\begin{equation}\label{x1u}
1-6ux_1(u)\not=0,\qquad |u|<u_c,
\end{equation}
hence $y(u)$ is analytic in the disk $|u|<u_c$.
This shows that for real $0\le u<u_c$, both $x=x_1(u)$ and $y(u)>0$ are real, so the support of the equilibrium measure
is the interval $[a,b]$ on the real axis. Moreover, the endpoints $z=a$ and $z=b$ are analytic functions
of the parameter $u$.

In order to prove that we are in the one-cut regular case, we also need to show that the double root $z_0$
lies outside $[a,b]$:

\begin{proposition}\label{suppeq}
For $0\leq u<u_c$, we have that $z_0>b$, so that the double root is outside the interval $[a,b]$.
\end{proposition}

\begin{proof}
By the second equation in \eqref{system},
\[
z_0=\frac{1}{3u}-x\to +\infty,
\]
as $u\to 0^+$, so $z_0>b$ for small $u>0$. Suppose that for some $0<u<u_c$ we have $z_0=b$. Then
\[
\frac{1}{3u}-x=x+y,
\]
hence
\[
\left(\frac{1}{3u}-2x\right)^2=y^2=\frac{4}{1-6ux}\,,
\]
and
\[
(1-6ux)^3=36u^2.
\]
Denote $v=ux$. Then the latter equation and \eqref{cubic} give two cubic equations on $v$:
\[
(1-6v)^3-36u^2=0,\qquad 18v^3-9v^2+v-6u^2=0.
\]
By applying the Euclidean algorithm to these two cubic equations, we obtain that $u^2(1-34992u^4)=0$,
which is not true for $0<u<u_c$. Therefore, $z_0>b$ for $0<u<u_c$.
\end{proof}

When $u=u_c$, the support of the equilibrium measure will be the interval
\begin{equation}
[a,b]=[3^{3/4}-3^{5/4},3^{3/4}+3^{1/4}],
\end{equation}
and the double root of the function $\varrho(z)^2$ is located at
\begin{equation}
z_c=3^{3/4}+3^{1/4},
\end{equation}
which coincides with the right endpoint of the support of the equilibrium measure.

We prove now that there exists an extension of the interval $J$ to a curve $\Gamma$ in $\mathbb{C}$
with right properties for the Riemann--Hilbert analysis.

\begin{proposition}
There exist two curves $\gamma_1$ and $\gamma_2$ in $\mathbb{C}$ such that conditions \eqref{Rephi} hold true.
\end{proposition}

\begin{center}
\begin{figure}[h]\label{Gamma_t1}
\begin{center}
\scalebox{0.7}{\includegraphics{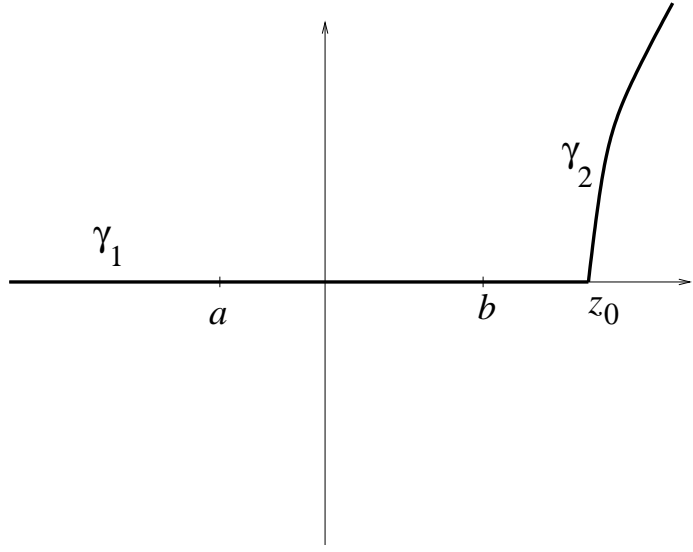}}
\end{center}
  \caption[sectors ]{The  contour  $\tilde \Ga_0=\ga_1\cup[a,b]\cup\ga_2$.}
 \end{figure}
\label{Gammat1}
\end{center}

\begin{proof}
If $b<z<z_0$ then $(z-a)(z-b)>0$ and $z-z_0<0$ again, so $\Re \phi_2(z)>0$ and we can take the segment $(b,z_0)$ as part of $\gamma_2$. However, we cannot take the whole interval $(b,\infty)$ as $\gamma_2$, since $\Re \phi_2(z)<0$ if $z>z_0$. Nevertheless, if we define the quadratic differential $-Q(z)dz^2$, where
\begin{equation}
Q(z)=\frac{1}{4}R(z)h^2(z),
\end{equation}
then it follows from the general theory (see \cite[\S 8.2]{Pom} or \cite[Theorem 7.1]{Str}) that, since $z=z_0$ is a double zero of $-Q(z)dz^2$, then there are exactly four trajectories that emanate from $z=z_0$, along which $\Im Q(z)=0$. Locally, these trajectories leave the point $z=z_0$ with angles $k\pi/2$, $0\leq k\leq 3$. We choose the trajectory that leaves $z=z_0$ with angle $\pi/2$ into the upper half plane. This trajectory must go to $\infty$, since $\phi_2(z_0)>0$ and $\phi_2(z)$ is increasing on $\gamma_2$, so it cannot go back to the real axis. Now,
\begin{equation}
\begin{aligned}
0<\phi_2(z)-\phi_2(z_0)&=\frac{1}{2}\int_{z_0}^z \sqrt{R(s)}h(s)ds\\
&=\int_{z_0}^z \left(\frac{V'(s)}{2}-\om(s)\right)ds\\
&=\frac{V(z)-V(z_0)}{2}+\mathcal{O}(\ln|z|)=
-\frac{uz^3}{3}+\mathcal{O}(z^2)
\end{aligned}
\end{equation}
when $|z|\to\infty$. Since $\phi_2(z)$ is real for $z\in\gamma_2$, we have
$\Re \phi_2(z)=\phi_2(z)$, and
\begin{equation}
\Re \phi_2(z)=-\frac{uz^3}{3}+\mathcal{O}(z^2)>0.
\end{equation}
This also shows that $\gamma_2$ goes to infinity with angle $\pi/3$.

We can take $\gamma_1=(-\infty,a)$, since for $z<a$ we have $(z-a)(z-b)>0$ and $z-z_0<0$, so
\begin{equation}
\frac{1}{2}\sqrt{R(z)}h(z)=\frac{3u}{2}\sqrt{(z-a)(z-b)}(z-z_0)<0,
\end{equation}
and hence $\Re\, \phi_1(z)>0$ for $z<a$, after integrating. Note that on $(-\infty,a)$ we need to take the square root with negative sign. On $(-\infty,a)$, with a similar computation as before, we have
\begin{equation}
\Re \phi_1(z)=\phi_1(z)=-\frac{uz^3}{3}+\mathcal{O}(z^2),
\end{equation}
and the result follows.
\end{proof}

\subsection{The Riemann--Hilbert analysis}

We will consider monic orthogonal polynomials (OPs) $P_n(z)=z^n+\ldots$ with respect to the weight function
\begin{equation}\label{w_f}
w(z)= e^{-NV(z)},
\end{equation}
so that
\begin{equation}\label{op_1}
\int_{\Gamma} P_n(z)z^k w(z)dz=0, \qquad k=0,1,\ldots,n-1.
\end{equation}
The existence of the OPs is not obvious, since the contour $\Ga$ is complex, but the Riemann--Hilbert (RH)
analysis provides us, among other properties, with the existence and uniqueness of $P_n(z)$ for large enough
values of $n$. The RH analysis also gives us various properties of the OPs. In particular, we will see that
the OPs satisfy the three-term recurrence relation,
\begin{equation}\label{TTRR}
zP_n(z)=P_{n+1}(z)+\beta_nP_n(z)+\gamma_n^2 P_{n-1}(z),
\end{equation}
and that the recurrence coefficients $\be_n$, $\ga_n^2$ satisfy string and deformation equations.

It is known from the work of Fokas, Its and Kitaev \cite{FIK}, that the family of orthogonal polynomials
can be characterized as the solution of the following $2\times 2$ Riemann--Hilbert problem (RHP):

Find
a $2\times 2$ matrix-valued $Y_n : \mathbb C \setminus (\Ga_0\cup\Ga_1) \to \mathbb{C}^{2\times 2}$ such that
\begin{itemize}
\item $Y_n(z)$ is analytic for $z \in \mathbb{C}\setminus (\Ga_0\cup\Ga_1)$ and the limits
\[
\lim_{s\to z\pm 0}Y_{n}(s)=Y_{n\pm}(z)
\]
exist as $s$ approaches $z$ from the $\pm$-side of $\Gamma_0\cup\Gamma_1$. As usual, we assume that the
$+$-side is on the left of an oriented contour.
\item For $z \in \Ga_0\cup\Ga_1$,
\begin{equation}\label{RHP_1}
Y_{n+}(z)=Y_{n-}(z) \begin{pmatrix} 1 & \al(z)e^{-NV(z)} \\ 0 & 1 \end{pmatrix},
\end{equation}
where
\begin{equation}\label{RHP_2}
\al(z)=\al\chi_{\Ga_0}(z)+(1-\al)\chi_{\Ga_1}(z),
\end{equation}
and $\chi_A(z)$ denotes the characteristic function of the set $A$.
\item  As $z\to\infty$,
\begin{equation}\label{RHP_3}
Y_n(z)=\left(I+\mathcal{O}\left(\frac 1z\right)\right)
    \begin{pmatrix} z^n & 0 \\ 0 & z^{-n} \end{pmatrix}.
\end{equation}
\end{itemize}
We will call $n$ the {\it degree} of the RHP. This RHP has a unique solution if and only if
the monic polynomial $P_n(z)$, orthogonal with respect to the
weight function $w(z)e^{-NV(z)}$, uniquely exists (see the next section). If additionally
$P_{n-1}(z)$ uniquely exists, then the solution of the Riemann--Hilbert
problem is given by:
\begin{equation}\label{Yz}
Y_n(z)= \begin{pmatrix} P_n(z)  & (\mathcal{C}P_n w)(z) \\[1mm]
    -\frac{2\pi i}{h_{n-1}}P_{n-1}(z) & -\frac{2\pi i}{h_{n-1}}(\mathcal{C}P_{n-1}w)(z)
    \end{pmatrix},
\end{equation}
where
\begin{equation}\label{Cauchy}
    (\mathcal{C}f)(z)=\frac{1}{2\pi i}
    \int_{\Gamma}\frac{f(s)}{s-z}\ d s
\end{equation}
is the Cauchy transform on $\Gamma$, and the coefficient
$h_{n-1}$ is defined as
\begin{equation}\label{hn}
h_{n-1}=\int_{\Gamma}P^2_{n-1}(s)w(s)d s.
\end{equation}

The nonlinear steepest descent method, due to Deift and Zhou (see \cite{DZ} and also \cite{DKMVZ}), consists of a series of explicit and invertible transformations of this original Riemann--Hilbert problem, with the objective of arriving at a transformed problem for a $2\times 2$ matrix $R(z)$, where the matrix jumps tend to the identity uniformly in $\mathbb{C}$. These steps are briefly as follows:
\begin{itemize}
 \item $Y_n\mapsto T_n$, normalization at $\infty$, using the equilibrium measure on $J$.
 \item $T_n\mapsto S_n$, opening of lenses around $J$ to convert the highly oscillatory jump on $J$ into a combination of exponentially close to identity and constant matrix jumps.
 \item $S_n\mapsto R_n$, construction of global and local parametrices, which are (respectively) approximations near the interval $J$ and close to the endpoints of it. In this case, the local approximation near the endpoints $z=a$ and $z=b$ are given in terms of Airy functions.
\end{itemize}

We refer the reader to the works \cite{DKMVZ}, \cite{Deift}, \cite{EML} for the details.
In this case, the Riemann-Hilbert analysis gives the following result:

\begin{theo} \label{RHP_t} For any $\de>0$, there exist $N_0(\de)>0$ and $\ep(\de)>0$ such that for any
$u\in[0,u_c-\de]$ and any $N\ge N_0(\de)$, a solution  $Y_n(z)$ to the RHP \eqref{RHP_1}-\eqref{RHP_3}
exists for $n$ in the interval
\begin{equation}\label{I_N}
I_N=\Big\{ n:\;1-\ep(\de)\le \frac{n}{N}<1+\ep(\de)\Big\}.
\end{equation}
\end{theo}

One of the key results also given by the Riemann--Hilbert analysis is that the
recurrence coefficients $\gamma_n^2$ and $\beta_n$ admit an asymptotic expansion in inverse powers of
$N$ for large $N$. We will return to these expansions later, in Section \ref{rec_coeff}, but first we need to
discuss various properties of the orthogonal polynomials, which follow from the RHP.

\section{Orthogonal polynomials $P_n(z)$}\label{OPs}

The orthogonality
conditions \eqref{op_1} are equivalent to a linear system of equations for the coefficients of
the orthogonal polynomial
\begin{equation}\label{op1}
P_n(z)=z^n+a_{n-1}z^{n-1}+\ldots+a_1 z+a_0.
\end{equation}
Namely, if
\begin{equation}\label{op2}
c_j=\int_\Ga z^jw(z)dz,
\end{equation}
then
\begin{equation}\label{op3}
\sum_{j=0}^{n-1} c_{j+k} a_j=-c_{n+k},\qquad k=0,1,\ldots, n-1.
\end{equation}
Therefore, the orthogonal polynomial $P_n(z)$ exists and it is unique if and only if linear system
\eqref{op3} is nondegenerate, so that
\begin{equation}\label{op4}
D_{n-1}=
\left|
\begin{matrix}
c_0 & c_1 & \ldots & c_{n-1} \\
c_1 & c_2 & \ldots & c_{n} \\
\vdots & \vdots & \ddots & \vdots \\
c_{n-1} & c_n & \ldots & c_{2n-2}
\end{matrix}
\right|\not=0.
\end{equation}
In this case,
\begin{equation}\label{op5}
P_n(z)=\frac{D_n(z)}{D_{n-1}}
\end{equation}
where
\begin{equation}\label{op5a}
D_n(z)=\left|
\begin{matrix}
c_0 & c_1 & \ldots & c_{n-1} & c_n\\
c_1 & c_2 & \ldots & c_{n} & c_{n+1}\\
\vdots & \vdots & \ddots & \vdots & \vdots \\
c_{n-1} & c_n & \ldots & c_{2n-2} & c_{2n-1}\\
1 & z & \ldots & z^{n-1} & z^n
\end{matrix}
\right|,
\end{equation}
see \cite{Sze}. In addition, we have

\begin{prop} \label{OPs1} If there exists a unique orthogonal polynomial $P_n(z)$ then
\begin{equation}\label{op6}
D_n=h_nD_{n-1},
\end{equation}
where
\begin{equation}\label{op7}
h_n=\int_{\Ga} P_n(z)^2 w(z)dz.
\end{equation}
\end{prop}

For a proof of this and subsequent propositions of this section see Appendix \ref{AOP} at the end of the paper.
Another useful formula for $D_{n-1}$ is
\begin{equation}\label{op10}
D_{n-1}=\frac{1}{n!}\int_\Ga\ldots\int_\Ga \De(z)^2\,
\prod_{j=1}^n e^{-N V(z_j)}dz_1\ldots dz_n,
\end{equation}
where
\begin{equation}\label{op11}
\De(z)=\prod_{1\leq j<k\leq n}(z_j-z_k)^2\,.
\end{equation}
see \cite{Sze}.

Now we will relate the existence of a solution to RHP \eqref{RHP_1}-\eqref{RHP_3} to
the existence and the uniqueness of orthogonal polynomials.

\begin{prop} \label{exi} Suppose that  RHP \eqref{RHP_1}-\eqref{RHP_3} has a solution $Y_n(z)$. Then
there exists a unique orthogonal polynomial $P_n(z)$.
\end{prop}

The following proposition proves the existence of the three term recurrence relation:

\begin{prop} \label{ttr}
Suppose that RHP \eqref{RHP_1}-\eqref{RHP_3} has solutions $Y_{n-1}(z)$, $Y_n(z)$, and $Y_{n+1}(z)$
for the degrees $n-1$, $n$, and $n+1$, respectively. Then
the orthogonal polynomials $P_{n-1}(z)$, $P_n(z)$, and $P_{n+1}(z)$, which uniquely exist
by Proposition \ref{exi}, satisfy the three term recurrence relation,
\begin{equation}\label{rr1}
zP_n(z)=P_{n+1}(z)+\beta_nP_n(z)+\gamma_n^2 P_{n-1}(z).
\end{equation}
\end{prop}

Observe that under the conditions of Proposition \ref{ttr}, $D_n\not=0$, $D_{n-1}\not=0$, and $D_{n-2}\not=0$ and
hence, by \eqref{op6},
\begin{equation}\label{rr4}
h_n=\frac{D_n}{D_{n-1}}\not=0,\qquad h_{n-1}=\frac{D_{n-1}}{D_{n-2}}\not=0.
\end{equation}
From \eqref{rr1} we obtain, by multiplying by $P_{n-1}(z)$ and integrating with respect to $w(z)$, that
\begin{equation}\label{rr5}
\ga_n^2=\frac{h_n}{h_{n-1}}\not=0\,.
\end{equation}

Note that the proof of Propositions \ref{OPs1}-\ref{ttr} does not use specific form \eqref{w_f} of the weight
$w(z)$, and therefore they hold in a much more general situation. In the two subsequent propositions we
assume that the weight $w(z)$ has  form \eqref{w_f}.

\begin{prop} \label{str_1}
Suppose that RHP \eqref{RHP_1}-\eqref{RHP_3} has solutions $Y_j(z)$
for the degrees $j=n,n\pm 1,n\pm 2$. Then the
string equations  hold:
\begin{equation}\label{string}
\begin{aligned}
3u(\gamma_{n+1}^2+\beta_n^2+\gamma_n^2) & =\be_n,\\
\gamma_n^2(1-3u(\beta_n+\beta_{n-1}))& =\frac{n}{N}\,.
\end{aligned}
\end{equation}
\end{prop}

By Theorem \ref{RHP_t}, for any $\de>0$,  Propositions \ref{OPs1}-\ref{str_1} are applied for $u\in[0,u_c-\de]$,
provided $\{n,n\pm 1,n\pm 2\}\subset I_N=I_N(\de)$. Observe that $I_N(\de)$ does not depend on $u$.

\begin{prop} \label{an_1} For any $\de>0$, if $\{n,n\pm 1,n\pm 2\}\subset I_N(\de)$,
then the coefficients $\be_n$
and $\ga_n^2$ are $C^\infty$-functions of $u$ for $u\in[0,u_c-\de]$ and they are analytic functions of
$u$ for $u\in(0,u_c-\de]$.
\end{prop}

{\it Remark.} In what follows we use the recurrence coefficient $\ga_n^2$, which is uniquely defined,
and we will not use $\ga_n$, which is defined up to a sign. If needed, the sign of $\ga_n$ can be uniquely
defined by the analytic continuation of $\ga_n$ in $u$ starting from the value $\ga_n(0)=\sqrt{n/N}$.
Note, that, in general, the recurrence coefficients $\be_n$ and $\ga_n^2$ are complex valued.

\section{String equations and the large $N$ expansion of the recurrence coefficients}\label{rec_coeff}

\subsection{The large $N$ expansion of the recurrence coefficients} From the
 Riemann--Hilbert analysis, we know that $\gamma^2_{n}$ and $\beta_n$ can be expanded in
inverse powers of $N$, as functions of the parameter $s=n/N$.
Moreover, in the expansion of $\ga_n^2$
the odd coefficients vanish, and we have an expansion
in inverse powers of $N^{-2}$. Similarly, for $\be_n$ we also have an expansion
in inverse powers of $N^{-2}$, if we consider $\be_n$ as a function of the
shifted parameter $s=\frac{n}{N}+\frac{1}{2N}$.
Namely, we have the following result:

\begin{theo} {\rm (See \cite{BI})}. For any $\de>0$, if $n\subset I_N(\de/2)$, then the coefficients $\be_n$, $\ga_n^2$ can be expanded in uniform asymptotic series,
\begin{equation}\label{asympgb}
\begin{aligned}
\gamma_n^2 & \sim \sum_{k=0}^{\infty}\frac{1}{N^{2k}}g_{2k}\left(\frac{n}{N},u\right),\\
\beta_n & \sim \sum_{k=0}^{\infty}\frac{1}{N^{2k}}
b_{2k}\left(\frac{n}{N}+\frac{1}{2N},u\right),
\end{aligned}
\end{equation}
where the functions $g_{2k}(s,u)$, $b_{2k}(s,u)$, $k=0,1,\ldots$, do not
depend on $n$ and $N$, and they are $C^\infty$-smooth in $s$ and $u$ on the set $\{ 0\le u\le u_c-\de, \;|s-1|\le\ep(\de)/2\}$.
\end{theo}

The proof of this theorem in \cite{BI} is based on the string equations and it is applied to the current
case without any changes.
To find the functions $g_{2k}(s,u)$,
$b_{2k}(s,u)$ iteratively, we
substitute expansions \eqref{asympgb} into string equations \eqref{string} and equate
coefficients at different powers of $N^{-2}$. Let us consider the first equation in  \eqref{string}.
If we take $s=\frac{n}{N}+\frac{1}{2N}$, then from the first equation in \eqref{asympgb} we have
the asymptotic expansions
\begin{equation}\label{asympgb2s}
\begin{aligned}
\gamma_n^2 & \sim \sum_{k=0}^{\infty}\frac{1}{N^{2k}}g_{2k}\left(s-\frac{1}{2N},u\right),\\
\gamma_{n+1}^2 & \sim \sum_{k=0}^{\infty}\frac{1}{N^{2k}}g_{2k}\left(s+\frac{1}{2N},u\right),
\qquad s=\frac{n}{N}+\frac{1}{2N}\,.
\end{aligned}
\end{equation}
We substitute these expansions into the first equation in  \eqref{string} and expand functions the $g_{2k}$
into Taylor series at $s$:
\begin{equation}\label{asympgb3}
3u\left[2\sum_{k=0}^{\infty}\sum_{j=0}^{\infty}\frac{g_{2k}^{(2j)}(s,u)}{(2j)!2^{2j}N^{2k+2j}}
+\left(\sum_{k=0}^{\infty}\frac{b_{2k}(s,u)}{N^{2k}}\right)^2\right]=\sum_{k=0}^{\infty}\frac{b_{2k}(s,u)}{N^{2k}}\,,
\end{equation}
where
\begin{equation}
g_{2k}^{(2j)}(s,u)=\frac{\partial^{2j}g_{2k}(s,u)}{\partial s^{2j}}\,.
\end{equation}
By equating the coefficients at the same powers of $N^{-2}$ on both sides,
we obtain a series of equations for the functions $g_{2k}$, $b_{2k}$.

Similarly, if  we take
$s=\frac{n}{N}$, then from the second equation in \eqref{asympgb} we obtain that
\begin{equation}\label{asympgb4}
\left(\sum_{k=0}^{\infty}\frac{g_{2k}(s,u)}{N^{2k}}\right)
\left[1-6u\left(\sum_{k=0}^{\infty}\sum_{j=0}^{\infty}\frac{b_{2k}^{(2j)}(s,u)}{(2j)!2^{2j}N^{2k+2j}}\right)\right]=s.
\end{equation}
Let us analyze the equations obtained when we equate  the coefficients at the same powers of $N^{-2}$ on both sides
of equations \eqref{asympgb3}, \eqref{asympgb4}.

\subsection{Zeroth order terms in the recurrence coefficients}
 At the zeroth order of $N^{-2}$ we obtain the system of equations,
\begin{align}\label{g0b0}
3u[2g_0(s,u)+b_0(s,u)^2] & = b_0(s,u)\,,\\
\nonumber
g_0(s,u)[1-6ub_0(s,u)] & = s.
\end{align}
These equations can be solved to yield
\begin{equation}\label{cubicg0s}
\begin{aligned}
72u^2 g_0(s,u)^3-g_0(s,u)^2+s^2&=0, \\
b_0(s,u)&=\frac{g_0(s,u)-s}{6ug_0(s,u)}\,,
\end{aligned}
\end{equation}
or alternatively,
\begin{equation}\label{cubicb0s}
\begin{aligned}
18b_0(s,u)^3u^2-9b_0(s,u)^2u+b_0(s,u)-6su&=0, \\
 g_0(s,u)&=\frac{s}{1-6ub_0(s,u)}\,.
\end{aligned}
\end{equation}
Consider $s$ in a small neighborhood of the point $s=1$ on the complex plane, say, $|s-1|\le c\ll 1$.
Then the cubic equation \eqref{cubicg0s} admits three different solutions, which behave as follows when $u\to 0$:
\begin{equation}
\begin{aligned}
g^{(1)}_0(s,u)&=s+\mathcal{O}(u^2)\\
g^{(2)}_0(s,u)&=-s+\mathcal{O}(u^2)\\
g^{(3)}_0(s,u)&=\frac{1}{72u^2}+\mathcal{O}(1).
\end{aligned}
\end{equation}
The solution we are interested in is $g_0(s,u)=g_0^{(1)}(s,u)$.
Since $g^{(1)}_0(s,u)$ is an isolated solution of the cubic equation, it is analytic at $u=0$.
In fact it is analytic in the disk $|u|<u_c s^{-1/2}$ on the complex plane.
In the following proposition we explicitly evaluate the Taylor series of $g_0(s,u)=g_0^{(1)}(s,u)$ at $u=0$.

\begin{proposition}\label{g0T}
The Taylor series
\begin{equation}\label{seriesg0s}
g_0(s,u)=\sum_{j=0}^{\infty} a_{2j}(s) u^{2j}
\end{equation}
converges in the disk $|u|<u_c s^{-1/2}$
and its coefficients are
\begin{equation}\label{coeffsg0s}
a_{2j}(s)=\frac{72^j s^{j+1}\Gamma(\frac{3j+1}{2})}{2\Gamma(j+1)\Gamma(\frac{j+3}{2})}, \qquad j\geq 0.
\end{equation}
\end{proposition}

\begin{proof}
We make the change of variables,
\begin{equation}
w=su^2,\qquad v=u^2 g_0(s,u).
\end{equation}
Then $v=v(w)$ satisfies the cubic equation
\begin{equation}\label{cubicv_1}
72v^3-v^2+w^2=0,
\end{equation}
and we are looking for the root $v=v(w)$ such that $v(0)=0$ and $v(w)>0$ for small $w>0$. Observe that $v(w)$ admits
an expansion in powers of $w$,
\begin{equation}\label{vcj_1}
v(w)=\sum_{j=1}^{\infty} c_j w^j.
\end{equation}
Using Cauchy integral formula, if $\gamma$ denotes a smooth closed contour around $w=0$, we have
\begin{equation}
c_j=\frac{1}{2\pi i}\oint_{\gamma} \frac{v(w)}{w^{j+1}}dw.
\end{equation}
Since $v(0)=0$, we make a change of variables
\begin{equation}
c_j=\frac{1}{2\pi i}\oint_{\gamma} \frac{v}{w(v)^{j+1}}\frac{dw}{dv}dv.
\end{equation}
We have
\begin{equation}
w(v)=\frac{v}{s}\sqrt{1-72v}, \qquad
\frac{dw}{dv}=\frac{1-108v}{\sqrt{1-72v}},
\end{equation}
so
\begin{equation}
c_j=\frac{1}{2\pi i}
\oint_{\gamma} v^{-j}(1-108v)(1-v)^{-\frac{j}{2}-1}dv.
\end{equation}
We expand the binomial series:
\begin{equation}
\begin{aligned}
 (1-108v)(1-v)^{-\frac{j}{2}-1}&=(1-108v)\sum_{k=0}^{\infty} {-\frac{j}{2}-1\choose k}(-1)^k (72v)^k\\
&=(1-108v)\sum_{k=0}^{\infty} {k+\frac{j}{2} \choose k} (72v)^k,
\end{aligned}
\end{equation}
and we pick up the residues, that correspond to $k=j-1$ in the first term and $k=j-2$ in the second one. Then
\begin{equation}\label{cj_1}
c_j=72^{j-1}\left[{\frac{3j}{2}-1 \choose j-1}-\frac{3}{2}{\frac{3j}{2}-2 \choose j-2}\right]
=\frac{\Gamma(\frac{3j}{2}-1)72^{j-1}}{2\Gamma(j)\Gamma(\frac{j}{2}+1)}
\end{equation}
Now,
\begin{equation}\label{seriesg0su}
\begin{aligned}
g_0(s,u) & =\frac{s}{w}v(w)=\frac{s}{w} \sum_{j=1}^{\infty} c_j w^j
=s\sum_{j=0}^{\infty}c_{j+1}w^j\\
& =\sum_{j=0}^{\infty}
\frac{72^j s^{j+1}\Gamma(\frac{3j+1}{2})u^{2j}}{2\Gamma(j+1)\Gamma(\frac{j+3}{2})} ,
\end{aligned}
\end{equation}
which gives the expression for the coefficients $a_{2j}(s)$.
\end{proof}

It is easy to generate several coefficients using general formula \eqref {coeffsg0s}:
\begin{equation}\label{g0su}
\begin{aligned}
g_0(s,u)=&s+36s^2u^2+3240s^3u^4+373248s^4u^6+48498912s^5u^8\\
&+\mathcal{O}(u^{10}).
\end{aligned}
\end{equation}
As a consequence of Proposition \ref{g0T},  we have

\begin{corollary}
The coefficient $b_0(s,u)$ is an analytic function of $u$ for $|u|<u_c s^{-1/2}$.
\end{corollary}

\begin{proof} From the second equation in \eqref{cubicg0s} and formula \eqref{g0su} it is clear that
$b_0(s,u)$ is analytic at $u=0$. In addition, $g_0(s,u)\not =0$ for $|u|<u_c s^{-1/2}$, because if $g_0(s,u)=0$
then by the first equation in \eqref{cubicg0s} $s=0$, but we assume that $|s-1|\le c\ll 1$. Therefore,
the analyticity of $g_0(s,u)$ implies the analyticity of $b_0(s,u)$ in the disk $|u|<u_c s^{-1/2}$.
\end{proof}

To obtain the asymptotics of $\ga_N^2$ and $\be_N$ as $N\to\infty$,
we set $n=N$ in \eqref{asympgb}, so that
\begin{align}\label{asympgb2}
\gamma_N^2 & \sim g_0(1,u)+\sum_{k=1}^{\infty}\frac{1}{N^{2k}}\,g_{2k}(1,u),\\
\beta_N & \sim b_0\left(1+\frac{1}{2N},u\right)+\sum_{k=1}^{\infty}\frac{1}{N^{2k}}\,
b_{2k}\left(1+\frac{1}{2N},u\right).
\end{align}
In the sequel, we denote
$$
g_{2k}(u)=g_{2k}(1,u), \qquad b_{2k}(u)=b_{2k}(1,u),
$$
for brevity.  When $s=1$, the cubic equation \eqref{cubicg0s} for $g_0(u)$ becomes
\begin{equation}\label{cubicg0}
72u^2 g_0(u)^3-g_0(u)^2+1=0,
\end{equation}
and the solution admits the power series expansion
\begin{equation}\label{seriesg01}
\begin{aligned}
g_0(u) & =\sum_{j=0}^{\infty}
\frac{72^j \Gamma(\frac{3j+1}{2})\, u^{2j}}{2\Gamma(j+1)\Gamma(\frac{j+3}{2})},
\end{aligned}
\end{equation}
valid when $|u|<u_c$. Formula \eqref{g0su} reduces for $s=1$ to
\begin{equation}\label{g0ser}
g_0(u)=1+36u^2+3240u^4+373248u^6+48498912u^8+\mathcal{O}(u^{10}).
\end{equation}

\subsection{Analysis of the higher order terms}
An important consequence of  string equations \eqref{asympgb3} and \eqref{asympgb4}
is that higher order terms $g_{2k}=g_{2k}(s,u)$ and $b_{2k}=b_{2k}(s,u)$, $k\geq 1$, can
be expressed in terms of $g_0(s,u)$, $b_0(s,u)$, and their derivatives
with respect to $s$. More precisely:
\begin{itemize}
 \item We can solve for $g_{2k}$ and $b_{2k}$ in terms of the previous coefficients through a {\emph{linear}} system of two equations.
 \item The analysis of the determinant of this linear system shows that higher order terms have the same singularity $u=\pm u_c$ as the initial terms $g_0$ and $b_0$, and that no other singularities appear when increasing $k$. In this respect, this property is analogous to the one presented in \cite{Erc}, \cite{EML} and \cite{EMLP} for general even potentials $V(M)$.
\item Although an explicit formula for $g_{2k}$ and $b_{2k}$ is rather complex when $k\geq 2$,
in order to determine the large $j$ asymptotic behavior of the coefficients $f^{(2g)}_{2j}$ in the expansion
of the free energy, it is enough to evaluate the leading order term at
the singular points $u=\pm u_c$, and this is possible to extract from the string equations.
\end{itemize}

For convenience, we make the change of variable
\begin{equation}
z=\frac{\xi}{u},
\end{equation}
so that the potential $V(z)$ becomes
\begin{equation}
\hat{V}(\xi)=\frac{1}{u^2}\left(\frac{\xi^2}{2}-\xi^3\right).
\end{equation}
With this change of variable, we introduce a new family of orthogonal polynomials
\begin{equation}
\hat{P}_n(\xi)=u^n P_n\left(\frac{\xi}{u}\right),
\end{equation}
and the corresponding recurrence coefficients are
\begin{equation}
\hat{\beta}_n=u\beta_n, \qquad \hat{\gamma}^2_n=u^2\gamma_n^2.
\end{equation}
Now the string equations \eqref{string} read
\begin{equation}\label{stringxi}
\begin{aligned}
3(\hat{\gamma}_{n+1}^2+\hat{\beta}_n^2+\hat{\gamma}_n^2)&=\hat{\beta}_n\\
\hat{\gamma}^2_n(1-3(\hat{\beta}_n+\hat{\beta}_{n-1}))&=\frac{nu^2}{N}\,.
\end{aligned}
\end{equation}
We will use the same scaled parameter as before,
\begin{equation}\label{w}
w=su^2.
\end{equation}
From \eqref{asympgb} and \eqref{stringxi}
we have an asymptotic expansion for these new coefficients
in powers of $\frac{u^2}{N}$:
\begin{equation}\label{asympgbhat}
\begin{aligned}
\hat{\gamma}_n^2& \sim \sum_{k=0}^{\infty}\frac{u^{4k}}{N^{2k}}\,\hat{g}_{2k}\left(\frac{n u^2}{N}\right),\\
\hat{\beta}_n & \sim \sum_{k=0}^{\infty}\frac{u^{4k}}{N^{2k}}\,\hat{b}_{2k}\left(\frac{n u^2}{N}+\frac{u^2}{2N}\right),
\end{aligned}
\end{equation}
where
\begin{equation}\label{hatbg_1}
\begin{aligned}
\hat{g}_{2k}(w)=u^{-4k+2}g_{2k}(s,u),\qquad
\hat{b}_{2k}(w)=u^{-4k+1}b_{2k}(s,u).
\end{aligned}
\end{equation}
For $s=1$ this reduces to
\begin{equation}\label{hatbg_2}
\begin{aligned}
\hat{g}_{2k}(u^2)=u^{-4k+2}g_{2k}(u),\qquad
\hat{b}_{2k}(u^2)=u^{-4k+1}b_{2k}(u),
\end{aligned}
\end{equation}
and in particular, when $k=0$, we have
\begin{equation}
\hat{g}_0(w)=u^2g_0(u), \qquad \hat{b}_0(w)=ub_0(u).
\end{equation}
It follows from \eqref{cubicg0s} that
\begin{equation}\label{p0q0}
\hat{b}_0(w)=\frac{\hat{g}_0(w)-w}{6\hat{g}_0(w)}\,,
\end{equation}
and also that $\hat{g}_0(w)$ satisfies the following cubic equation:
\begin{equation}\label{cubicg0hat}
72\hat{g}_0^3(w)-\hat{g}_0^2(w)+w^2=0,
\end{equation}
which is identical to \eqref{cubicv_1}. The critical value now becomes
\begin{equation}\label{wc}
w_c=u_c^2=\frac{\sqrt{3}}{324}.
\end{equation}
By \eqref{vcj_1} and \eqref{cj_1}, the function $\hat{g}_0(w)$ is analytic at $w=0$,
and it has the following Taylor expansion:
\begin{equation}\label{g0hatT_1}
\hat{g}_0(w)=\sum_{j=1}^{\infty} \frac{\Gamma(\frac{3j}{2}-1)72^{j-1}w^j}{2\Gamma(j)\Gamma(\frac{j}{2}+1)}\,.
\end{equation}
From equations \eqref{asympgb3} and \eqref{asympgb4} we obtain the string equations for $\hat{g}_{2k}(w)$
and $\hat{b}_{2k}(w)$:
\begin{equation}\label{hat1}
\begin{aligned}
6\sum_{m+j=k}\frac{\hat g_{2m}^{(2j)}(w)}{(2j)!2^{2j}}
+3\sum_{m+m'=k}\hat b_{2m}(w)\hat b_{2m'}(w)
=\hat b_{2k}(w)
\end{aligned}
\end{equation}
and
\begin{equation}\label{hat2}
\hat g_{2k}(w)-6\sum_{m+m'+j=k}\frac{\hat g_{2m}(w)\hat b_{2m'}^{(2j)}(w)}{(2j)!2^{2j}}=0,\quad k\ge 1.
\end{equation}
The advantage of equations \eqref{p0q0}--\eqref{hat2} is that they do not contain the parameter $u$ anymore:
it is hidden in the argument $w=su^2$.

By solving these equations for $\hat{g}_{2k}(w)$ and $\hat{b}_{2k}(w)$, we obtain the following linear system of equations for $k\ge 1$:
\begin{equation}\label{hat3}
\begin{aligned}
&6\hat g_{2k}(w)+(6\hat b_0(w)-1)\hat b_{2k}(w)\\
&=-6\mathop{\sum_{m+j=k}}_{m\le k-1}\frac{\hat g_{2m}^{(2j)}(w)}{(2j)!2^{2j}}
-3\mathop{\sum_{m+m'=k}}_{m,m'\le k-1}\hat b_{2m}(w)\hat b_{2m'}(w)\,,
\end{aligned}
\end{equation}
and
\begin{equation}\label{hat4}
\begin{aligned}
(1-6\hat b_0(w))&\hat g_{2k}(w)-6\hat g_0(w)\hat b_{2k}(w)=6\mathop{\sum_{m+m'+j=k}}_{m,m'\le k-1}\frac{\hat g_{2m}(w)\hat b_{2m'}^{(2j)}(w)}{(2j)!2^{2j}}
\end{aligned}
\end{equation}
 The determinant of this system is
\begin{equation}\label{det0}
D(w)=-36\hat g_0(w)+(1-6\hat b_0(w))^2.
\end{equation}
By using equations \eqref{p0q0}, \eqref{cubicg0hat}, it can be simplified to
\begin{equation}\label{determinant}
D(w)=1-108 \hat g_0(w).
\end{equation}
This leads to the following result:

\begin{proposition}\label{propgb}
 For $k\geq 1$ the functions $\hat{g}_{2k}(w)$ and $\hat{b}_{2k}(w)$ are analytic functions
of $w$  in the disk $|w|<w_c$  in the complex plane.
\end{proposition}

\begin{proof}
The proof of this proposition is by induction in $k$. To this end, we need to show that
the determinant $D(w)$ does not vanish for
$|w|<w_c$.

Suppose $D(w)=0$ for some $w$ such that $|w|<w_c$. Then by \eqref{determinant}, we have
\begin{equation}\label{hot7}
\hat{g}_0(w)=\frac{1}{108}\,.
\end{equation}
By substituting this into the cubic equation \eqref{cubicg0hat} we obtain that
\begin{equation}\label{hot8}
w^2=\frac{1}{3\cdot 108^2}=w_c^2,
\end{equation}
hence $D(w)\not=0$ inside the disk $|w|<w_c$.

The proposition now follows, since it is clear that the only singularities of the higher order terms
for $|w|<w_c$ are those inherited from $\hat{g}_0(w)$ and $\hat{b}_0(w)$, but $\hat{g}_0(w)$ and $\hat{b}_0(w)$
are analytic in the disk $|w|<w_c$.
\end{proof}

As a consequence of Proposition \ref{propgb}, we have the following result:
\begin{corollary}\label{parity}
The functions $g_{2k}(s,u)$ and $b_{2k}(s,u)$ are analytic in the disk $|u|<u_c\,s^{-1/2}$ in the complex plane, and $g_{2k}(s,u)$ is even in $u$ and $b_{2k}(s,u)$ is odd in $u$.
\end{corollary}
\begin{proof}
Analyticity follows from Proposition \ref{propgb} and the change of variable $w=su^2$, and the parity, from rewriting equation \eqref{hatbg_1}:
\begin{equation}
\begin{aligned}
g_{2k}(s,u)=u^{4k-2}\hat{g}_{2k}(w),\qquad
b_{2k}(s,u)=u^{4k-1}\hat{b}_{2k}(w).
\end{aligned}
\end{equation}

\end{proof}


\section{The large $N$ expansion of the free energy $F_N$}\label{free_energy}

\subsection{Another change of variable and an analysis of the auxiliary potential}
To evaluate the large  $N$ expansion of the free energy $F_N$, we will use the Toda equation, which relates $F_N$ to the recurrence coefficient $\ga_N^2$.
To this end we will consider a different form of the partition function.
Under the change of variable
\begin{equation}\label{wz}
z=(3u)^{-1/3}\zeta+\frac{1}{6u}\,,
\end{equation}
where we assume that $u>0$ and $(3u)^{-1/3}>0$, we have, by  straightforward algebra, that
\begin{equation}\label{wz1}
\frac{z^2}{2}-uz^3-\frac{1}{108u^2}=-\frac{\zeta^3}{3}+t\zeta,
\end{equation}
where
\begin{equation}\label{tu}
t=\frac{1}{4(3u)^{4/3}}\,.
\end{equation}
Let us denote
\begin{equation}\label{ZN2}
\tilde{Z}_N=\int_{\tilde \Ga}\ldots\int_{\tilde \Ga} \prod_{1\leq j<k\leq N} (\zeta_j-\zeta_k)^2
\prod_{j=1}^N e^{-N\tilde{V}(\zeta_j)}d\zeta_1\ldots d\zeta_N,
\end{equation}
where
\begin{equation}
\tilde{V}(\zeta)=-\frac{\zeta^3}{3}+t\zeta,
\end{equation}
and $\tilde \Ga$ is the image of $\Ga$ under the change of variable \eqref{wz}.
Note that this cubic model was used in \cite{DHK} with $t=0$, in the context of the complex Gaussian quadrature
of integrals with high order stationary points. The corresponding free energy is
\begin{equation}
\tilde{F}_N=\frac{1}{N^2}\ln \tilde{Z}_N.
\end{equation}
Formula \eqref{wz1} implies that
\begin{equation}\label{FNFN}
\tilde{F}_N=\frac{1}{108 u^2}+\frac{\ln \,(3u)}{3}+F_N=\frac{2\,t^{3/2}}{3}-\frac{\ln\,(4t)}{4}+F_N\,.
\end{equation}
By Theorem \ref{Th1}, for any $\ep>0$, $F_N(u)$ admits a uniform asymptotic expansion for $0\le u\le u_c-\ep$,
hence for any $\ep>0$, $\tilde F_N(t)$ admits a uniform asymptotic expansion as well, for $t_c+\ep\le t< \infty$:
\begin{equation}\label{FNFN1}
\tilde{F}_N(t)\sim \sum_{k=0}^{\infty}\frac{ \tilde F^{(2k)}(t)}{N^{2g}}.
\end{equation}
where
\begin{equation}\label{FNFN0}
\tilde{F}^{(0)}(t)=\frac{1}{108 u^2}+\frac{\ln\, (3u)}{3}+F^{(0)}(u)=\frac{2\,t^{3/2}}{3}-\frac{\ln\,(4t)}{4}+F^{(0)}(u).
\end{equation}
and for $k\ge 1$,
\begin{equation}\label{FNFN2}
\tilde{F}^{(2k)}(t)=F^{(2k)}(u),\qquad u=\frac{1}{3(4t)^{3/4}}\,.
\end{equation}
From \eqref{uc} and \eqref{tu} we obtain that the new critical value is
\begin{equation}\label{tc0}
t_c=\frac{1}{4(3u_c)^{4/3}}=3\cdot 2^{-2/3}.
\end{equation}

Note that with the change of variable \eqref{wz}, the family of orthogonal polynomials is scaled and shifted.
If we write $z=c\zeta+d$, then:
\begin{equation}
P_n(z)=P_n(c\zeta+d)=c^n \tilde{P}_n(\zeta),
\end{equation}
so that $\tilde{P}_n(\zeta)$ is again monic. The coefficients of the recurrence relation are modified in the following way:
\begin{equation}\label{bg}
\tilde{\beta}_n=\frac{\beta_n-d}{c} , \qquad
\tilde{\gamma}_n^2=\frac{\gamma^2_n}{c^2}.
\end{equation}
In our case, according to \eqref{wc}, $c=(3u)^{-1/3}$ and $d=\frac{1}{6u}$, so
\begin{equation}\label{gbu}
\begin{aligned}
\tilde{\gamma}_n^2&=(3u)^{2/3}\gamma_n^2=\frac{\gamma_n^2}{2\sqrt{t}}\,,\\
\tilde{\beta}_n&=(3u)^{1/3}\left(\beta_n-\frac{1}{6u}\right)=\frac{1}{(4t)^{1/4}}\left(\beta_n-\frac{1}{6u}\right).
\end{aligned}
\end{equation}

An important fact in our analysis will be that with respect to the variable $t$ we have the following Toda equation
(see, e.g., \cite{BI}):
\begin{equation}\label{Toda}
\frac{d^2 \tilde{F}_N}{dt^2}=\tilde{\gamma}^2_N.
\end{equation}
The usual derivation of the Toda equation assumes the existence of the orthogonal polynomials $\tilde P_n(\z)$ for
$n=0,1,\ldots,N$. In our case we know the existence of the orthogonal polynomials $P_n(z)$, and hence
$\tilde P_n(\z)$ only for $n\in I_N$. We will prove that, nevertheless,
the Toda equation is valid in our case.
Namely, we will prove the following proposition:

\begin{prop} \label{Toda_eq} Toda equation \eqref{Toda} holds for any $t>t_c$.
\end{prop}

For a proof of this proposition see Appendix \ref{A_Toda} in the end of the paper.
The idea now is to integrate twice equation \eqref{Toda} in the variable $t$,
in order to obtain $\tilde{F}_N$, and then compute $F_N$ using equation \eqref{FNFN}.
From the asymptotic expansion \eqref{asympgb2}, we have
\begin{equation}\label{ggbu}
\tilde{\gamma}^2_N(t)
\sim \sum_{k=0}^{\infty} \frac{\tilde{g}_{2k}(t)}{N^{2k}}
=\frac{1}{2\sqrt{t}}\sum_{k=0}^{\infty} \frac{g_{2k}(u)}{N^{2k}},\qquad u=\frac{1}{3(4t)^{3/4}}\,.
\end{equation}
Integrating this expression twice should give an expansion for $\tilde{F}_N$, and then we can compute $F_N$ using \eqref{FNFN}. However, the problem now is that we have to justify that the term-by-term integration
of the large $N$ asymptotic expansion of $\tilde{\gamma}_N$ over an unbounded interval in the variable $t$ is permissible.
This is the content of the following proposition:

\begin{proposition} \label {prop_41} We have that
\begin{equation}\label{tF0}
\tilde F^{(0)}(t)=\frac{2\,t^{3/2}}{3}-\frac{\ln\,(4t)}{4}+
\int_{\infty}^t\int_\infty^\tau \left(\tilde g_0(\sg)-\frac{1}{2\sqrt \sg}-\frac{1}{4\sg^2}\right)
 d\sg d\tau,
\end{equation}
and for any $k\ge 1$,
\begin{equation}\label{tFk}
\tilde F^{(2k)}(t)=\int_{\infty}^t\int_\infty^\tau \tilde g_{2k}(\sg) d\sg d\tau.
\end{equation}
\end{proposition}

\begin{proof}
Let us prove first that  for any $k\ge 0$,
\begin{equation}\label{dF2}
\frac{d^2\tilde F^{(2k)}(t)}{dt^2}= \tilde g_{2k}(t).
\end{equation}
Consider the difference operator
\begin{equation}
\Delta_h f(t)=f(t+h)-2f(t)+f(t-h),
\end{equation}
where $h>0$ and $t$ are fixed. From the Taylor theorem with integral remainder, we have
\begin{equation}\label{dh0}
\Delta_h f(t)=\int_{t-h}^{t+h} f''(\tau)m_h(\tau,t)d\tau,
\end{equation}
where $m_h(\tau,t)$ is the hat-shaped function
\begin{equation}
m_h(\tau,t)=\left\{
\begin{array}{ll}
h-|\tau-t|, & |\tau-t|\leq h\\
0, & \textrm{otherwise}.
\end{array}\right.
\end{equation}
It is clear that $m_h(\tau,t)\geq 0$ for $\tau\in [t-h,t+h]$. Now, for any $K>0$,
$\ep>0$, $1\ge h >0$, and $t\ge t_c+1+\ep$,
\begin{equation}\label{dh1}
\begin{aligned}
\Delta_h \tilde{F}_N(t)&=\int_{t-h}^{t+h} \tilde{F}_N''(\tau)m_h(\tau,t)d\tau\\
&=\int_{t-h}^{t+h} \tilde{\gamma}_N(\tau)m_h(\tau,t)d\tau\\
&=\sum_{k=0}^K \int_{t-h}^{t+h} \frac{\tilde{g}_{2k}(\tau)}{N^{2k}}m_h(\tau,t)d\tau+\mathcal{O}(N^{-2K-2}),
\end{aligned}
\end{equation}
where the error term is uniform in $t\ge t_c+1+\ep$, since the asymptotic expansion \eqref{ggbu}
is uniform in $t$. On the other hand,
\begin{equation}\label{dh2}
\Delta_h \tilde{F}_N(t)=\sum_{k=0}^K  \frac{\Delta_h\tilde{F}^{(2k)}(t)}{N^{2k}}
+\mathcal{O}(N^{-2K-2}),
\end{equation}
where the error term is uniform in $t\ge t_c+1+\ep$, because the asymptotic expansion \eqref{FNFN1}
is uniform in $t$.
Since the coefficients of an asymptotic series in powers of $N^{-2}$ for the function $\Delta_h \tilde{F}_N(t)$
are uniquely determined, we obtain from \eqref{dh1}, \eqref{dh2} that for any $k\ge 0$,
\begin{equation}
 \Delta_h\tilde{F}^{(2k)}(t)=\int_{t-h}^{t+h} \tilde{g}_{2k}(\tau)m_h(\tau,t)d\tau.
\end{equation}
But by \eqref{dh0},
\begin{equation}
\Delta_h \tilde{F}^{(2k)}(t)=\int_{t-h}^{t+h} \frac{d^2 \tilde{F}^{(2k)}(\tau)}{d\tau^2}\,m_h(\tau,t)d\tau,
\end{equation}
hence
\begin{equation}
\int_{t-h}^{t+h} \frac{d^2 \tilde{F}^{(2k)}(\tau)}{d\tau^2}\,m_h(\tau,t)d\tau=
\int_{t-h}^{t+h} \tilde{g}_{2k}(\tau)m_h(\tau,t)d\tau,
\end{equation}
and since $h$ and $t$ are arbitrary and $m_h(\tau,t)$ is a positive function, equation \eqref{dF2} holds.

By Corollary \ref{parity}, the function $g_{2k}(u)$ is an even function of $u$, analytic at $u=0$, hence as $u\to 0$,
\begin{equation}\label{g2ku}
g_{2k}(u)=a_0(k)+a_2(k)u^2+\mathcal{O}(u^{4}).
\end{equation}
Because of \eqref{ggbu}, we know that
\begin{equation}\label{gtg_1}
\tilde{g}_{2k}(t)=\frac{g_{2k}(u)}{2\sqrt t}, \qquad k\geq 0,
\end{equation}
and this implies that as $t\to\infty$,
\begin{equation}
 \tilde{g}_{2k}(t)=\frac{g_{2k}(u)}{2\sqrt t}
=\frac{a_0(k)}{2\sqrt t}+\frac{a_1(k)}{72t^2}+\mathcal{O}(t^{-7/2}).
\end{equation}
Integrating twice we obtain that
\begin{equation}\label{CDR1}
 \tilde F^{(2k)}(t)=\frac{4a_0(k) t^{3/2}}{3}-\frac{a_1(k)\ln t}{72}+C+Dt+R_{2k}(t),
\end{equation}
where
\begin{equation}
R_{2k}(t)=\mathcal{O}(t^{-3/2})
\end{equation}
and $C,\,D$ are some unknown constants.
The error term can be written as
\begin{equation}\label{CDR1a}
 R_{2k}(t)=\int_{\infty}^t\int_\infty^\tau \left(\tilde g_{2k}(\sg)-\frac{a_0(k)}{2\sqrt \sg}-\frac{a_1(k)}{72\sg^2}\right)
 d\sg d\tau
\end{equation}
Combining \eqref{CDR1} for $k=0$ with \eqref{FNFN0} we obtain that
\begin{equation}\label{CDR2}
\begin{aligned}
 \tilde F^{(0)}(t)&=\frac{2\,t^{3/2}}{3}-\frac{\ln\,(4t)}{4}+F^{(0)}(u)\\
&=\frac{4a_0(0) t^{3/2}}{3}
-\frac{a_1(0)\ln t}{72}+C+Dt+R_{0}(t).
\end{aligned}
\end{equation}
Since
\begin{equation}\label{CDR3}
\lim_{u\to 0}F^{(0)}(u)=0,\qquad \lim_{t\to\infty}R_{0}(t)=0,
\end{equation}
we conclude from \eqref{CDR2}, by taking $t\to\infty$, that
\begin{equation}\label{CDR4}
a_0(0)=\frac{1}{2}\,,\quad a_1(0)=18,\quad C=-\frac{\ln 2}{2}\,,\quad D=0,
\end{equation}
and \eqref{tF0} follows from \eqref{CDR1a}, \eqref{CDR2}.

Similarly, combining \eqref{CDR1} for $k\ge 1$ with \eqref{FNFN2} we obtain that
\begin{equation}\label{CDR5}
\begin{aligned}
 \tilde F^{(2k)}(t)&=F^{(2k)}(u)\\
&=\frac{4a_0(0) t^{3/2}}{3}
-\frac{a_1(0)\ln t}{72}+C+Dt+R_{2k}(t).
\end{aligned}
\end{equation}
Again, since
\begin{equation}\label{CDR6}
\lim_{u\to 0}F^{(2k)}(u)=0,\qquad \lim_{t\to\infty}R_{2k}(t)=0,
\end{equation}
we conclude from \eqref{CDR5}, by taking $t\to\infty$, that
\begin{equation}\label{CDR7}
a_0(k)=0\,,\quad a_1(k)=0,\quad C=0\,,\quad D=0,
\end{equation}
and \eqref{tFk} follows from \eqref{CDR1a}, \eqref{CDR5}.
\end{proof}

\begin{remark} From \eqref{g2ku} and \eqref{CDR7} we obtain that if $k\ge 1$ then
\begin{equation}\label{g2ku2}
g_{2k}(u)=\mathcal{O}(u^{4}),\qquad u\to 0.
\end{equation}
\end{remark}

Also, from \eqref{FNFN0}, \eqref{FNFN2}, \eqref{tF0}, and \eqref{tFk},
we have the following corollary of Proposition \ref{prop_41}.

\begin{cor}
\begin{equation}\label{CDR8}
F^{(0)}(u)=
\int_{\infty}^t\int_\infty^\tau \left(\tilde g_0(\sg)-\frac{1}{2\sqrt \sg}-\frac{1}{4\sg^2}\right)
 d\sg d\tau,
\end{equation}
and
\begin{equation}\label{CDR9}
F^{(2k)}(u)=\int_{\infty}^t\int_\infty^\tau \tilde g_{2k}(\sg) d\sg d\tau,
\qquad t=\frac{1}{4(3u)^{4/3}}\,;\qquad k\ge 1.
\end{equation}
\end{cor}

Using this result, we will prove Theorem \ref{Th2}, by integrating the explicit expression that we found for the leading coefficient $g_0(u)$. Also, we will prove Theorem \ref{Th3} by using an expression for $\hat{g}_2(w)$ obtained from the string equations.

\section{Proof of Theorem \ref{Th2} and Theorem \ref{Th3}}
\label{Th2Th3}

Observe that from \eqref{tu} we have that
\begin{equation}\label{CDR10}
72u^2=t^{-3/2},
\end{equation}
so \eqref{seriesg01} implies that
\begin{equation}
g_0(u)=\sum_{j=0}^{\infty}
\frac{72^j \Gamma(\frac{3j+1}{2})\, u^{2j}}{2\Gamma(j+1)\Gamma(\frac{j+3}{2})}
=\sum_{j=0}^{\infty} \frac{\Gamma(\frac{3j+1}{2}) t^{-3j/2}}{2\Gamma(j+1)\Gamma(\frac{j+3}{2})}\,.
\end{equation}
Multiplying by the factor $(3u)^{2/3}=\frac{1}{2\sqrt{t}}$, we obtain
\begin{equation}
\begin{aligned}
\tilde{g}_0(t) &= \frac{1}{2\sqrt{t}}+\frac{1}{4t^2}+
\sum_{j=2}^{\infty} \frac{\Gamma(\frac{3j+1}{2})t^{-\frac{3j+1}{2}}}{4\Gamma(j+1)\Gamma(\frac{j+3}{2})}\,.
\end{aligned}
\end{equation}
Integrating twice in $t$, according to formula \eqref{CDR8}, we get
\begin{equation}\label{seriesf0}
\begin{aligned}
F^{(0)}(u)=\sum_{j=2}^{\infty} \frac{\Gamma(\frac{3j+1}{2})t^{-\frac{3j-3}{2}}}{(3j-1)(3j-3)\Gamma(j+1)\Gamma(\frac{j+3}{2})},
\end{aligned}
\end{equation}
and making the change of variable $t=\frac{1}{4(3u)^{4/3}}$, we have
\begin{equation}\label{seriesf0u}
\begin{aligned}
F^{(0)}(u)&=
\sum_{j=2}^{\infty} \frac{3^{2(j-1)} 4^{\frac{3j-3}{2}}
\Gamma(\frac{3j+1}{2})u^{2(j-1)}}{(3j-1)(3j-3)\Gamma(j+1)\Gamma(\frac{j+3}{2})}\\
&=\sum_{j=1}^{\infty} \frac{72^j
\Gamma(\frac{3j}{2})u^{2j}}{2\Gamma(j+3)\Gamma(\frac{j}{2}+1)}.
\end{aligned}
\end{equation}
This proves formula \eqref{Th2_1}.

 From the explicit expression of the coefficients $f^{(0)}_{2j}$ in formula \eqref{Th2_1} we obtain the asymptotic formula,
\begin{equation}
f^{(0)}_{2j}=\frac{(2j)!}{\sqrt{6\pi}\,u_c^{2j}\, j^{7/2}}
\left(1-\frac{115}{36j}+\frac{19705}{2592j^2}+\mathcal{O}(j^{-3})\right), \qquad j\to
\infty,
\end{equation}
which gives \eqref{asympf02j}, together with some higher order corrections.
Theorem \ref{Th2} is proved.

When $k=1$, system \eqref{hat3}--\eqref{hat4} reads
\begin{equation}\label{systemg2b2}
\begin{aligned}
&
\begin{pmatrix}
6 & 6\hat{b}_0(w)-1\\
1-6\hat{b}_0(w) & -6\hat{g}_0(w)
\end{pmatrix}
\begin{pmatrix}
\hat{g}_2(w)\\
\hat{b}_2(w)
\end{pmatrix}
=\frac{3}{4}
\begin{pmatrix}
-\hat{g}_0''(w)\\
\hat{g}_0(w)\hat{b}_0''(w)
\end{pmatrix}.
\end{aligned}
\end{equation}
This system can be solved (and simplified, using the string equations) to give
\begin{align}
\hat{g}_2(w)&=\frac{162\hat{g}_0(w)(5-324\hat{g}_0(w))}{(1-108\hat{g}_0(w))^4},\label{g2b2g} \\
\hat{b}_2(w)&=\frac{54w}{\hat{g}_0(w)(1-108\hat{g}_0(w))^4}\label{g2b2b}.\\
\end{align}
An explicit expression \eqref{Th3_1} for the coefficients $f^{(2)}_{2j}$ can be now obtained from equation
\eqref{g2b2g}. Namely, $\hat{g}_2(w)$ can be expanded in powers of $w$ around the origin,
\begin{equation}
\hat{g}_2(w)=\sum_{j=2}^{\infty} c_j w^j,
\end{equation}
starting with a $w^2$ term. From cubic equation \eqref{cubicg0hat}, we obtain
\begin{equation}
w=\hat{g}_0(w)\sqrt{1-72\hat{g}_0(w)},
\end{equation}
so writing $v=\hat{g}_0(w)$ and applying the Cauchy integral formula, we find
\begin{equation}
\begin{aligned}
c_j&=\frac{1}{2\pi i}\oint_{\gamma} \frac{(5-324v)v^{-j}(1-72v)^{\frac{-j-1}{2}}}{162(1-108v)^4}\left(\sqrt{1-72v}-\frac{36v}{\sqrt{1-72v}}\right)dv,\\
&=\frac{1}{2\pi i}\oint_{\gamma} \frac{(5-324v)v^{-j}(1-72v)^{-\frac{j}{2}}}{162(1-108v)^3}dv,
\end{aligned}
\end{equation}
where $\gamma$ is a smooth closed contour around the origin in the $v$ plane. Now we expand the binomial series and pick up the residue at $v=0$, to get
\begin{equation}
c_j=162\cdot 72^{j-1}\sum_{m=0}^{j-1}{\frac{3j}{2}-m-1 \choose j-m-1}(m+1)(m+5)\left(\frac{3}{2}\right)^m.
\end{equation}
If we write this in terms of the standard Pochhammer symbol, see \cite{dlmf}, we see that
\begin{equation}
(m+1)(m+5)=5\,\frac{(2)_m (6)_m}{(1)_m (5)_m},
\end{equation}
and also
\begin{equation}
{\frac{3j}{2}-m-1 \choose j-m-1}
={\frac{3j}{2}-1 \choose j-1}\frac{(-j+1)_m}{\left(-\frac{3j}{2}+1\right)_m},
\end{equation}
so we can identify $c_j$ with the following hypergeometric function:
\begin{equation}
c_j=\frac{45\cdot 72^j}{4} {\frac{3j}{2}-1 \choose j-1}
\, _3F_2\left(\begin{array}{l} -j+1,2,6\\ 5,-\frac{3j}{2}+1 \end{array};\frac{3}{2}\right).
\end{equation}
Since one of the parameters in the numerator equals another one in the denominator plus one, we can simplify this $_3F_2$ function in terms of Gauss hypergeometric functions using the following straightforward identity:
\begin{equation}\label{F32F21}
_3F_2\left(\begin{array}{l} a,b,c+1\\ c,d\end{array};z\right)
=\, _2F_1\left(\begin{array}{l} a,b\\ d\end{array};z\right)
+\frac{abz}{cd}\, _2F_1\left(\begin{array}{l} a+1,b+1\\ d+1\end{array};z\right),
\end{equation}
and thus
\begin{equation}\label{F21}
\begin{aligned}
c_j&=
\,\frac{45}{4} 72^j {\frac{3j}{2}-1 \choose j-1}\times\\
&\left[_2F_1\left(\begin{array}{l} -j+1,2\\ -\frac{3j}{2}+1 \end{array};\frac{3}{2}\right)
+\frac{6(j-1)}{5(3j-2)}\, _2F_1\left(\begin{array}{l} -j+2,3\\ -\frac{3j}{2}+2 \end{array};\frac{3}{2}\right)\right].
\end{aligned}
\end{equation}
Next, observe from \eqref{hatbg_1} and \eqref{hatbg_2} that
\begin{equation}
\tilde{g}_2(t)=\frac{1}{2\sqrt{t}}g_2(u), \qquad
g_{2}(u)=u^2 \hat{g}_2(w),
\end{equation}
so
\begin{equation}
\begin{aligned}
\tilde{g}_2(t)&=\frac{1}{2\sqrt{t}}g_2(u)=\frac{1}{144t^2}\hat{g}_2(w)
=\frac{1}{144t^2}\hat{g}_2\left(\frac{1}{72t^{3/2}}\right)\\
&=\frac{1}{144t^2}\sum_{j=1}^{\infty}c_j\left(\frac{1}{72t^{3/2}}\right)^j.
\end{aligned}
\end{equation}
Using formula \eqref{CDR9}, we integrate twice in $t$, and going back to the variable $u$, we find
\begin{equation}\label{F2u}
F^{(2)}(u)=\frac{1}{36}\sum_{j=1}^{\infty}\frac{c_ju^{2j}}{3j(3j+2)}\,,
\end{equation}
so the coefficients are
\begin{equation}
\begin{aligned}
f^{(2)}_{2j}&= \frac{c_j(2j)!}{108j(3j+2)}\\
&=\frac{45\cdot 72^j}{4} {\frac{3j}{2}-1 \choose j-1}
\, _3F_2\left(\begin{array}{l} -j+1,2,6\\ 5,-\frac{3j}{2}+1 \end{array};\frac{3}{2}\right)
 \frac{(2j)!}{108j(3j+2)}\\
&=\frac{5\cdot 72^j \Gamma\left(\frac{3j}{2}\right)(2j)!}{48(3j+2)\Gamma(j+1)\Gamma\left(\frac{j}{2}+1\right)}
 \, _3F_2\left(\begin{array}{l} -j+1,2,6\\ 5,-\frac{3j}{2}+1 \end{array};\frac{3}{2}\right).
\end{aligned}
\end{equation}
This proves formula \eqref{Th3_1}. The large $j$ asymptotic behavior of these coefficients is given in Appendix \ref{App_Th3_2} below, proving formula  \eqref{Th3_2}.

\section{Case of genus $g>1$}\label{Th4}

The expressions for the terms in the topological expansion become more and more complicated as $k$ grows.
Fortunately, as we pointed out before, in order to understand the large $j$ behavior of the terms $f^{(2g)}_{2j}$,
it is enough to consider the leading order terms near the critical point $w=w_c$ only.
We recall that $w_c=\frac{\sqrt{3}}{324}$, as given in \eqref{wc}.

When $w=w_c$, the cubic equation $72\hat{g}_0^3(w)-\hat{g}_0^2(w)+w^2=0$ given in \eqref{cubicg0hat} has three roots,
\begin{equation}
r_1=r_2=\frac{1}{108}, \qquad r_3=-\frac{1}{216}.
\end{equation}
Substituting
\begin{equation}
w=w_c-\Delta w,\qquad \hat{g}_0(w)=\frac{1}{108}+\Delta\hat{g}_0(w)
\end{equation}
into the cubic equation, we have
\begin{equation}
72(\Delta\hat{g}_0(w))^2\left(\hat{g}_0(w)+\frac{1}{216}\right)=w_c^2-w^2
=\Delta w (w_c+w),
\end{equation}
Therefore,
\begin{equation}
\begin{aligned}
\Delta\hat{g}_0&=\pm \sqrt{\frac{(w+w_c)\Delta w}{72\left(\hat{g}_0(w)+\frac{1}{216}\right)}}=\pm \sqrt{\frac{2w_c\Delta w}{72\left(\hat{g}_0(w_c)+\frac{1}{216}\right)}}+\mathcal{O}(\Delta w)\\
&=\pm \frac{2^{1/2}3^{1/4}}{18} (\Delta w)^{1/2}+\mathcal{O}(\Delta w)
\end{aligned}
\end{equation}
So, near the critical value $w=w_c$, our solution satisfies
\begin{equation}\label{g0critical}
\hat{g}_0(w)=\frac{1}{108}-\frac{2^{1/2}3^{1/4}}{18}(\Delta w)^{1/2}+\mathcal{O}(\Delta w),
\end{equation}
taking the minus sign since the function $\hat{g}_0(w)$ is increasing with $w$ and therefore $\hat{g}_0(w)<\hat{g}(w_c)$.

Since we have $\hat{b}_0(w)$ explicitly in terms of $\hat{g}_0(w)$, see formula \eqref{p0q0}, we can obtain the asymptotic behavior of $\hat{b}_0(w)$ as well:
\begin{equation}\label{b0critical}
\hat{b}_0(w)=\frac{3-\sqrt{3}}{18}-2^{1/2}3^{-1/4}(\Delta w)^{1/2}+\mathcal{O}(\Delta w).
\end{equation}
We also note that near the critical point $w=w_c$, the determinant $D(w)$ in \eqref{determinant} behaves as follows:
\begin{equation}\label{detcritical}
 D(w)=2^{3/2}3^{5/4}(\Delta w)^{1/2}+\mathcal{O}(\Delta w)=D'(\Delta w)^{1/2}+\mathcal{O}(\Delta w).
\end{equation}
We can write formulae \eqref{g0critical}, \eqref{b0critical} as
\begin{equation}\label{criticalg0b0}
\begin{aligned}
\hat{g}_0(w)&=\hat{g}_0(w_c)+C_0(\Delta w)^{1/2}+\mathcal{O}(\Delta w),\\
\hat{b}_0(w)&=\hat{b}_0(w_c)+D_0(\Delta w)^{1/2}+\mathcal{O}(\Delta w),\\
\end{aligned}
\end{equation}
where
\begin{equation}\label{hbcr}
\hat{g}_0(w_c)=\frac{1}{108}\,,\qquad \hat{b}_0(w_c)=\frac{3-\sqrt{3}}{18}\,,
\end{equation}
and
\begin{equation}\label{C0D0}
C_0=-\frac{2^{1/2}3^{1/4}}{18}\,,\qquad D_0=-2^{1/2}3^{-1/4}.
\end{equation}
Observe that
\begin{equation}\label{D0C0}
D_0=6\sqrt{3}\,C_0.
\end{equation}
For higher order terms, $k\geq 1$, we make the following Ansatz:
\begin{equation}\label{ansatz}
\begin{aligned}
\hat{g}_{2k}(w)&=C_{2k}(\Delta w)^{\frac{1}{2}-\frac{5k}{2}}+\mathcal{O}((\Delta w)^{1-\frac{5k}{2}}),\\
\hat{b}_{2k}(w)&=D_{2k}(\Delta w)^{\frac{1}{2}-\frac{5k}{2}}+\mathcal{O}((\Delta w)^{1-\frac{5k}{2}}),
\end{aligned}
\end{equation}
where the functions $\hat{g}_{2k}(w)$, $\hat{b}_{2k}(w)$ have a square root singularity in the complex
plane at the point $w=w_c$.
We will prove this Ansatz by induction with respect to $k$, by using the system of equations \eqref{hat3}--\eqref{hat4}.  Simultaneously, we will derive recurrence equations
for $C_{2k},\;D_{2k}$.

Let us analyze the first sum on the right hand side of \eqref{hat3}.
By differentiating the first formula in \eqref{ansatz}, we obtain
\begin{equation}\label{ans1}
 \hat{g}_{2m}^{(2j)}(w)=C_{2m}^{(2j)}(\Delta w)^{\frac{1}{2}-\frac{5m}{2}-2j}+\ldots, \qquad
m\le k-1,
\end{equation}
where the dots indicate higher order terms with respect to $\Delta w$ and
\begin{equation}
C_{2m}^{(2j)}=C_{2m}\left(\frac{5m}{2}-\frac{1}{2}\right)\ldots
\left(\frac{5m}{2}+(2j-1)-\frac{1}{2}\right).
\end{equation}
Since $m+j=k$,
we can write formula \eqref{ans1} as
\begin{equation}\label{ans2}
 \hat{g}_{2m}^{(2j)}(w)=C_{2m}^{(2j)}(\Delta w)^{\frac{1}{2}-\frac{5k}{2}+\frac{j}{2}}+\ldots,
\qquad j\ge 1.
\end{equation}
Hence the leading term in  the first sum on the right hand side of \eqref{hat3} corresponds
to $j=1$, $m=k-1$, so
\begin{equation}\label{ans3}
\begin{aligned}
&-6\sum_{m+j=k;\; m\le k-1}\frac{\hat g_{2m}^{(2j)}(w)}{(2j)!2^{2j}}\\
&=-\frac{3C_{2k-2}}{4}
\left(\frac{5(k-1)}{2}-\frac{1}{2}\right)
\left(\frac{5(k-1)}{2}+\frac{1}{2}\right)(\Delta w)^{1-\frac{5k}{2}}+\ldots\\
&=-\frac{3C_{2k-2}\left(5k-6\right)\left(5k-4\right)}{16}
(\Delta w)^{1-\frac{5k}{2}}+\ldots\\
\end{aligned}
\end{equation}
In the second sum, we have
\begin{equation}\label{ans4}
\begin{aligned}
-3\mathop{\sum_{m+m'=k}}_{m,m'\le k-1}\hat b_{2m}(w)\hat b_{2m'}(w)
=-3\mathop{\sum_{m+m'=k}}_{m,m'\le k-1} D_{2m}D_{2m'}(\Delta w)^{1-\frac{5k}{2}}+\ldots
\end{aligned}
\end{equation}
For $k=1$ this sum is equal to 0. Substituting formulae \eqref{ans3}, \eqref{ans4} into \eqref{hat3}, we obtain
\begin{equation}\label{ans5}
6\hat g_{2k}(w)-(1-6\hat b_0(w))\hat b_{2k}(w)
=A_{2k}(\Delta w)^{1-\frac{5k}{2}}+\ldots,
\end{equation}
where
\begin{equation}\label{ans6}
A_{2k}=-\frac{3C_{2k-2}\left(5k-6\right)\left(5k-4\right)}{16}-3\sum_{m+m'=k;\; m,m'\le k-1} D_{2m}D_{2m'}.
\end{equation}

Similarly, in equation \eqref{hat4} the leading singular term on the right, with $(\Delta w)^{1-\frac{5k}{2}}$,
corresponds to $j=1,\;m=0,\;m'=k-1$ and to $j=0,\; m+m'=k$, so that
\begin{equation}\label{ans7}
\begin{aligned}
(1-6\hat b_0(w))\hat g_{2k}(w)-6\hat g_0(w)\hat b_{2k}(w)=B_{2k}(\Delta w)^{1-\frac{5k}{2}}+\ldots,
\end{aligned}
\end{equation}
where
\begin{equation}\label{ans8}
B_{2k}=\frac{D_{2k-2}\left(5k-6\right)\left(5k-4\right)}{576}+6\sum_{m+m'=k;\; m,m'\le k-1} C_{2m}D_{2m'}.
\end{equation}
Solving system \eqref{ans5}, \eqref{ans7} and using asymptotic formula \eqref{detcritical}
for its determinant, we obtain that
\begin{equation}\label{ans9}
\begin{aligned}
\hat{g}_{2k}(w)&=C_{2k}(\Delta w)^{\frac{1}{2}-\frac{5k}{2}}+\mathcal{O}((\Delta w)^{1-\frac{5k}{2}}),\\
\hat{b}_{2k}(w)&=D_{2k}(\Delta w)^{\frac{1}{2}-\frac{5k}{2}}+\mathcal{O}((\Delta w)^{1-\frac{5k}{2}}),
\end{aligned}
\end{equation}
where
\begin{equation}\label{ans10}
\begin{aligned}
C_{2k}&=\frac{1}{2^{3/2}3^{5/4}}\bigg[-6\hat g_0(w_c)A_{2k}+(1-6 \hat b_0(w_c))B_{2k}\bigg],\\
D_{2k}&=\frac{1}{2^{3/2}3^{5/4}}\bigg[-(1-6 \hat b_0(w_c))A_{2k}+6B_{2k}\bigg].
\end{aligned}
\end{equation}
This proves Ansatz \eqref{ansatz} and gives the recurrence formula for the coefficients $C_{2k},\;D_{2k}$.
Moreover, substituting formulae \eqref{hbcr}, we obtain that
\begin{equation}\label{ans11}
\begin{aligned}
C_{2k}&=\frac{1}{2^{3/2}3^{5/4}}\bigg(-\frac{A_{2k}}{18}+\frac{\sqrt {3}\,B_{2k}}{3}\bigg),\\
D_{2k}&=\frac{1}{2^{3/2}3^{5/4}}\bigg(-\frac{\sqrt {3}\,A_{2k}}{3}+6B_{2k}\bigg).
\end{aligned}
\end{equation}
Hence
\begin{equation}\label{ans12}
D_{2k}=6\sqrt{3}\, C_{2k},
\end{equation}
which extends relation \eqref{D0C0} to $k\ge 1$. By using this relation, we can eliminate $D_{2m}$ from
formulae \eqref{ans6}, \eqref{ans8}:
\begin{equation}\label{ans13}
\begin{aligned}
A_{2k}&=-\frac{3C_{2k-2}\left(5k-6\right)\left(5k-4\right)}{16}-324\sum_{m+m'=k;\; m,m'\le k-1} C_{2m}C_{2m'},\\
B_{2k}&=\frac{\sqrt{3}\,C_{2k-2}\left(5k-6\right)\left(5k-4\right)}{96}+36\sqrt{3}\sum_{m+m'=k;\; m,m'\le k-1} C_{2m}C_{2m'},
\end{aligned}
\end{equation}
and hence by the first formula in \eqref{ans11},
\begin{equation}\label{ans14}
\begin{aligned}
C_{2k}&=\frac{1}{2^{3/2}3^{5/4}}\bigg(\left(5k-6\right)\left(5k-4\right)\frac{C_{2k-2}}{48}+54\mathop{\sum_{m+m'=k}}_{m,m'\le k-1} C_{2m}C_{2m'}\bigg).
\end{aligned}
\end{equation}
This proves recursive formula \eqref{gg3}.

From \eqref{C0D0} and \eqref{gg3}, we find the values of the constant $C_{2k}$ for $k=0,1,2$ as
\begin{equation}\label{ans14a}
\begin{aligned}
C_0=-\frac{2^{1/2}3^{1/4}}{18}\,,\qquad C_2=\frac{1}{48\cdot 108}\,,
\qquad C_4=\frac{49\cdot 2^{1/2}3^{3/4}}{17915904}\,.
\end{aligned}
\end{equation}

Let us apply formula \eqref{ansatz} to find the asymptotic behavior of $F^{(2k)}(u)$ at $u=u_c$ for $k\ge 1$. From \eqref{hatbg_2} and \eqref{gtg_1}, we know that
\begin{equation}
\tilde{g}_{2k}(t)=\frac{g_{2k}(u)}{2\sqrt{t}}= \frac{u^{4k-2}\hat{g}_{2k}(w)}{2\sqrt{t}}
=\frac{36t\hat{g}_{2k}(w)}{(72^2 t^3)^k},
\end{equation}
and then we obtain from \eqref{ans9} that
\begin{equation}\label{ans15}
\tilde{g}_{2k}(t)=\tilde{C}_{2k}(\Delta t)^{\frac{1}{2}-\frac{5k}{2}}+\ldots,
\end{equation}
where
\begin{equation}\label{ans16}
\tilde{C}_{2k}=\frac{36 t_c C_{2k}}{(72^2 t_c^3)^k}\, \left(-\frac{dw}{dt}\Big\vert_{t=t_c}\right)^{\frac{1}{2}-\frac{5k}{2}}.
\end{equation}
If we integrate \eqref{ans15} twice with respect to $\Delta t$, we get
\begin{equation}
\tilde{F}^{(2k)}(t)= \frac{4 \tilde C_{2k}}{(5k-5)(5k-3)}
(\Delta t)^{\frac{5}{2}-\frac{5}{2}k}+\ldots,
\end{equation}
if $k\ge 2$. By \eqref{FNFN2} we obtain now that
\begin{equation}\label{ans17}
F^{(2k)}(u)= A_{2k}\left(1-\frac{u^2}{u_c^2}\right)^{\frac{5}{2}-\frac{5}{2}k}+\ldots, \qquad k\ge 2,
\end{equation}
where
\begin{equation}\label{ans18}
A_{2k}= \frac{4 \tilde C_{2k}u_c^{5-5k}}{(5k-5)(5k-3)}
\left(-\frac{dt}{dw}\Big\vert_{w=w_c}\right)^{\frac{5}{2}-\frac{5}{2}k}.
\end{equation}
Observe that by \eqref{CDR10},
\begin{equation}\label{ans19}
w=\frac{1}{72 t^{3/2}},\qquad t_c=3\cdot 2^{-2/3},\qquad -\frac{dw}{dt}\Big\vert_{t=t_c}=\frac{1}{48 t_c^{5/2}}\,,
\end{equation}
hence
\begin{equation}\label{ans20}
A_{2k}= \frac{144\cdot 48^2\, t_c^{6} u_c^5 }{(5k-5)(5k-3)}\left(\frac{1}{72^2 t_c^3u_c^5}\right)^kC_{2k}.
\end{equation}
Substituting expression \eqref{ans16} and simplifying, we obtain
\begin{equation}\label{ans21}
A_{2k}= \frac{24\cdot 3^{1/4}C_{2k} }{(5k-5)(5k-3)u_c^k}\,.
\end{equation}
Using the binomial expansion, we have
\begin{equation}\label{ans22}
\left(1-\frac{u^2}{u_c^2}\right)^{\frac{5}{2}-\frac{5}{2}k}
=\sum_{j=0}^\infty \frac{c_{j}u^{2j}}{u_c^{2j}}\,,
\end{equation}
where $c_j$ has the following asymptotic behavior as $j\to\infty$:
\begin{equation}\label{ans23}
c_j=\frac{j^{(5k-7)/2}}{\Gamma\left(\frac{5k-5}{2}\right)}\left(1+\mathcal O(j^{-1})\right)\,.
\end{equation}
Combining \eqref{ans17}, \eqref{ans21}--\eqref{ans23}, we find
\begin{equation}\label{ans23a}
F^{(2k)}(u)= \sum_{j=1}^\infty d^{(2k)}_{2j} \frac{u^{2j}}{u_c^{2j}},
\end{equation}
where as $j\to\infty$,
\begin{equation}\label{ans23b}
d^{(2k)}_{2j}=K_{2k}j^{(5k-7)/2} \left(1+\mathcal O(j^{-1/2})\right),
\end{equation}
with
\begin{equation}\label{ans24}
K_{2k}=\frac{6\cdot 3^{1/4}C_{2k}}{\Gamma\left(\frac{5k-1}{2}\right)u_c^{k}}\,.
\end{equation}
This proves Theorem \ref{Thgg1}.

The initial values of $K_{2k}$ are
\begin{equation}\label{ans25}
K_{0}=\frac{1}{\sqrt{6\pi}}\,,\qquad K_2=\frac{1}{48}\,,\qquad K_4= \frac{7}{1440\sqrt{6\pi}}\,.
\end{equation}
When $k=1$, formula \eqref{ans17} becomes
\begin{equation}\label{ans26}
F^{(2)}(u)= -\frac{1}{48}\,\ln \left(1-\frac{u^2}{u_c^2}\right)+\ldots
=\sum_{j=0}^\infty \frac{c_{j}u^{2j}}{u_c^{2j}}\,,
\end{equation}
where
\begin{equation}\label{ans27}
c_j=\frac{1}{48j}\left(1+\mathcal O(j^{-1/2})\right)\,.
\end{equation}

\section{Counting 3-valent graphs on a Riemannian surface}
\label{graphs}

In this final section, we prove that the coefficient $f^{(2g)}_{p}$ in series \eqref{top2} is equal to the number of 3-valent
graphs with $p$ vertices on a closed Riemannian surface of genus $g$. By differentiating formula \eqref{ZN}
$p$ times with respect to $u$ and evaluating at the point $u=0$, we obtain
\begin{equation}\label{3v1}
\frac{Z_N^{(p)}(0)}{Z_N(0)}= N^p\left\langle \left(\sum_{i=1}^N z_i^3\right)^{p}\right\rangle_0 \,,
\end{equation}
where
\begin{equation}\label{3v2}
\begin{aligned}
\left\langle f(z_1,\ldots,z_N)\right\rangle_0 &=\frac{1}{Z_N(0)}
\int_\Ga\ldots\int_\Ga f(z_1,\ldots,z_N)\\
&\times\prod_{1\leq j<k\leq N}(z_j-z_k)^2\,
\prod_{j=1}^N e^{-\frac{Nz_j^2}{2}}dz_1\ldots dz_N
\end{aligned}
\end{equation}
is the mathematical expectation with respect to the Gaussian ensemble. By the Cauchy theorem,
we can deform the contour of integration $\Ga$ in
the integral on the right in \eqref{3v1}
to the real axis, and then we can return to the matrix integral:
\begin{equation}\label{3v3}
\begin{aligned}
\left\langle \left(\sum_{i=1}^N z_i^3\right)^{p}\right\rangle_0
&=\left\langle \left(\Tr\, M^3\right)^{p}\right\rangle_0
=\frac{1}{\tilde Z_N(0)}
\int_{\mathcal H_N} (\Tr\, M^3)^p e^{-N\Tr \frac{M^2}{2}}dM.
\end{aligned}
\end{equation}
For odd $p$ the latter integral is equal to 0.
To evaluate the integral for even $p=2q$, we apply the Wick theorem.

We have
\begin{equation}\label{wi1}
\Tr\, M^3=\sum_{i,j,k=1}^N M_{ij}M_{jk} M_{ki}
\end{equation}
and
\begin{equation}\label{wi2}
(\Tr\, M^3)^p=\sum_{i_1,j_1,k_1,\ldots,i_p,j_p,k_p=1}^N M_{i_1j_1}M_{j_1k_1} M_{k_1i_1}\ldots
M_{i_pj_p}M_{j_pk_p} M_{k_pi_p}.
\end{equation}
The covariance matrix of $M_{ij}$'s is
\begin{equation}\label{wi3}
\left\langle M_{ij} M_{kl}\right\rangle_0=\frac{\de_{il}\de_{jk}}{N}\,.
\end{equation}
By the Wick theorem, the expectation of the product of matrix entries with respect to the Gaussian measure can be expressed in terms of the product of expectations of all possible pairings of the matrix entries (see for instance \cite[\S 1.6]{For} for a more detailed exposition). In this case,
\begin{equation}\label{wi4}
\begin{aligned}
&\left\langle M_{i_1j_1}M_{j_1k_1} M_{k_1i_1}\ldots
M_{i_pj_p}M_{j_pk_p} M_{k_pi_p}\right\rangle_0\\
&=\sum_{\pi} \prod_{s=1}^{3p/2}\left\langle M_{i^{(s)}j^{(s)}} M_{k^{(s)}l^{(s)}}\right\rangle_0
=\sum_{\pi} N^{-3p/2}\prod_{s=1}^{3p/2}\de_{i^{(s)}l^{(s)}}\de_{j^{(s)}k^{(s)}}\,.
\end{aligned}
\end{equation}
where the sum is taken over all partitions
\begin{equation}\label{wi5}
\begin{aligned}
\pi&:\;\{M_{i_1j_1},M_{j_1k_1}, M_{k_1i_1},\ldots,
M_{i_pj_p},M_{j_pk_p}, M_{k_pi_p}\}\\
&=\bigsqcup_{s=1}^{3p/2}
 \{M_{i^{(s)}j^{(s)}}, M_{k^{(s)}l^{(s)}}\}
\end{aligned}
\end{equation}
of the set $\{M_{i_1j_1},M_{j_1k_1}, M_{k_1i_1},\ldots,
M_{i_pj_p},M_{j_pk_p}, M_{k_pi_p}\}$ into disjoint
pairs $\{M_{i^{(s)}j^{(s)}}, M_{k^{(s)}l^{(s)}}\}$.

From \eqref{wi2},\eqref{wi4} we have that
\begin{equation}\label{wi6}
\begin{aligned}
\left\langle(\Tr\, M^3)^p\right\rangle
&=\sum_{\pi}\sum_{i_1,j_1,k_1,\ldots,i_p,j_p,k_p=1}^N
 N^{-3p/2}\prod_{s=1}^{3p/2}\de_{i^{(s)}l^{(s)}}\de_{j^{(s)}k^{(s)}}\\
&=\sum_{\pi}
 N^{f-3p/2}\,,
\end{aligned}
\end{equation}
where $f$ is the number of cycles in the set $$\{M_{i_1j_1},M_{j_1k_1}, M_{k_1i_1},\ldots,
M_{i_pj_p},M_{j_pk_p}, M_{k_pi_p}\}$$ generated by the equalities
\begin{equation}\label{wi7}
i^{(s)}=l^{(s)},\quad j^{(s)}=k^{(s)},\quad s=1,\ldots, 3p/2.
\end{equation}
Thus,
\begin{equation}\label{wi8}
\frac{Z_N^{(p)}(0)}{Z_N(0)}= \sum_{\pi}
 N^{p+f-l},\qquad l=\frac{3p}{2}\,.
\end{equation}
By the second Wick theorem, we obtain now that
\begin{equation}\label{wi9}
F_N^{(p)}(0)= \frac{1}{N^2}\ln \frac{Z_N^{(p)}(0)}{Z_N(0)} =
{\sum_{\pi}}^c
 N^{p+f-l-2},
\end{equation}
where the sum is taken over partitions $\pi$ such that the graph $\Ga(\pi)$ is connected (for a general formulation of this result in combinatorics see for instance \cite[Vol. 2]{Sta}).

From \eqref{top1}, \eqref{top2} we have
\begin{equation}\label{wi10}
F_N^{(p)}(0)\sim \sum_{g\ge 0} \frac{f^{(2g)}_{p}}{N^{2g}}\,.
\end{equation}
Comparing this with \eqref{wi9}, we obtain that for each partition $\pi$,
\begin{equation}\label{wi11}
p+f-l-2=-2g
\end{equation}
and
\begin{equation}\label{wi12}
f^{(2g)}_{p}= \mathop{{\sum}^c}_{\pi:\; p+f-l-2=-2g} 1,
\end{equation}
As was noticed in \cite{BIZ}, by the Euler formula, the number $g$ in \eqref{wi11} is equal
to the minimal genus of a closed oriented Riemannian manifold
on which the graph $\Ga(\pi)$ can be realized without self-intersections. See the work \cite{Mul}
of Mulase for a rigorous proof. This proves that $f^{(2g)}_{p}$ counts the number of 3-valent graphs with $p$ vertices on
a closed oriented Riemannian manifold of genus $g$.
\begin{figure}
\scalebox{0.9}{\includegraphics{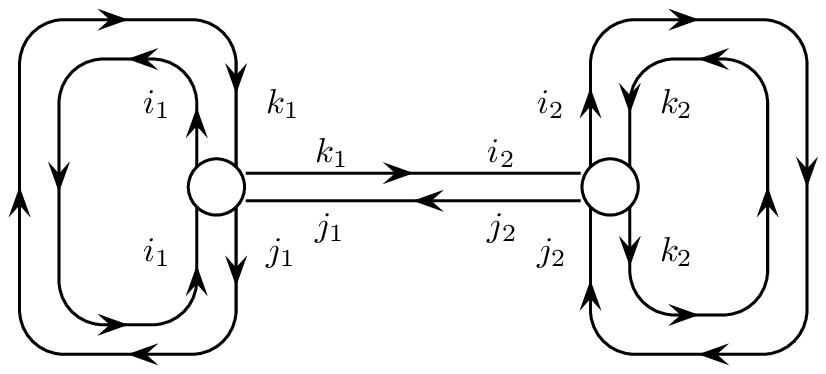}}
\scalebox{0.9}{\includegraphics{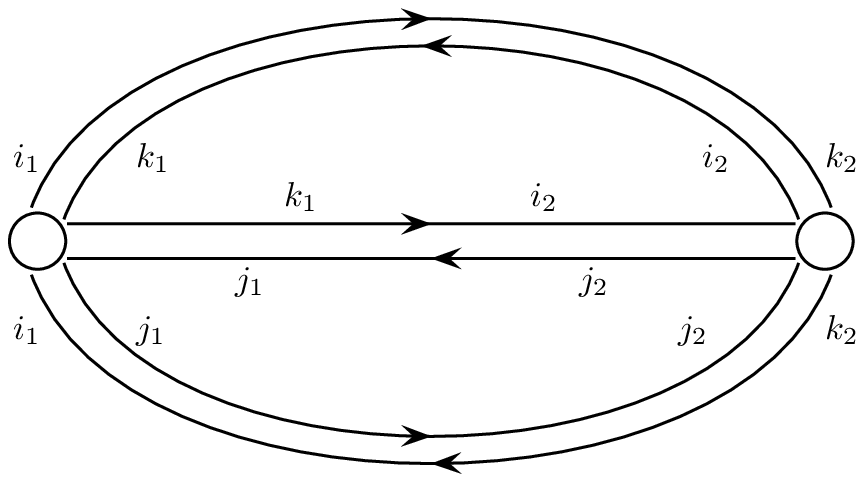}}
\caption{$3$-valent connected ribbon graphs with two vertices, genus $0$.}
\label{Figfatgraph1}
\end{figure}
\begin{figure}
\scalebox{0.9}{\includegraphics{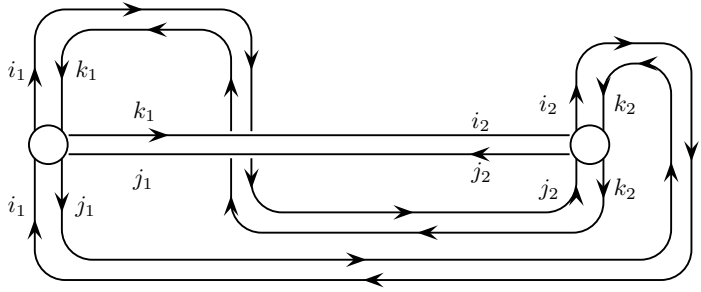}}
\caption{A $3$-valent connected fat graph with two vertices, genus $1$.}
\label{Figfatgraph3}
\end{figure}

As an illustration, Fig.2 depict possible 3-valent (ribbon) graphs with two vertices on the sphere, and Fig.3, on the torus.
Observe that by counting the possible enumerations of the edges, we obtain the multiplicity factor of the first graph
in Fig.2 to be equal to 9, and of the second graph to be equal to 3. The total number of the graphs with multiplicities
is equal to 12, which fits well to the first coefficient $6=\frac{12}{2}$ in the expansion of $F^{(0)}(u)$ given in \eqref{seriesF0}.
Similarly, we obtain that the multiplicity factor of the graph in Fig.3 on the torus is equal to 3, which
fits well to the first coefficient $\frac{3}{2}$ in the expansion of $F^{(2)}(u)$ presented in  \eqref{F2}.

\medskip

\appendix

\section{RHP and properties of the orthogonal polynomials $P_n(z)$}\label{AOP}

In this appendix we prove various propositions from Section \ref{OPs}.

{\it Proof of Proposition \ref{OPs1}}.
Applying equations \eqref{op3} to the last column in the determinant
for $D_n$, we obtain that
\begin{equation}\label{aop1}
\begin{aligned}
D_{n}&=
\left|
\begin{matrix}
c_0 & c_1 & \ldots & c_{n-1} & c_n\\
c_1 & c_2 & \ldots & c_{n} & c_{n+1}\\
\vdots & \vdots & \ddots & \vdots & \vdots \\
c_{n-1} & c_n & \ldots & c_{2n-2} & c_{2n-1}\\
c_n & c_{n+1} & \ldots & c_{2n-1} & c_{2n}
\end{matrix}
\right|\\
&=
\left|
\begin{matrix}
c_0 & c_1 & \ldots & c_{n-1} & 0 \\
c_1 & c_2 & \ldots & c_{n} & 0 \\
\vdots & \vdots & \ddots & \vdots & \vdots \\
c_{n-1} & c_n & \ldots & c_{2n-2} & 0 \\
c_n & c_{n+1} & \ldots & c_{2n-1} & h_n
\end{matrix}
\right|
=h_nD_{n-1},
\end{aligned}
\end{equation}
where
\begin{equation}\label{aop2}
h_n=\int_{\Ga} z^n P_n(z) e^{-NV(z)}dz=\int_{\Ga} P_n(z)^2 e^{-NV(z)}dz,
\end{equation}
which proves \eqref{op7}.\qed

{\it Proof of Proposition \ref{exi}.} Let $y_n(z)=\det Y_n(z)$. Then by \eqref{RHP_1},
 $y_{n+}(z)=y_{n-}(z)$, $z\in\Ga_0\cup\Ga_1$,
hence $y_n(z)$ is an entire function. By \eqref{RHP_2}, $y_n(z)\to 1$ as $z\to\infty$. Hence, by the Liouville
theorem, $y_n(z)\equiv 1$.

Next, we claim that $Y_n(z)$ is the unique solution of the RHP.
Indeed, suppose that $Y_n^{1,2}(z)$ are two solutions. Then
$X(z)=Y_n^1(z)Y_n^2(z)^{-1}$ has no jump on $\Ga_1\cup\Ga_2$ and $X(z)\to I$ as $z\to\infty$.
By the Liouville theorem, $X(z)\equiv 1$, hence $Y_n^{1}(z)=Y_n^{2}(z)$, so the solution is unique.

Consider now the (11)-element $Y_{n11}(z)$ of the matrix $Y_n(z)$. From \eqref{RHP_1} we obtain that
$Y_{n11+}(z)=Y_{n11-}(z)$, hence $Y_{n11}(z)$ is an entire function. By \eqref{RHP_3},
$Y_{n11}(z)=z^n+\mathcal O(z^{n-1})$, hence $Y_{n11}(z)$ is a monic polynomial of degree $n$.

Consider then the (12)-element $Y_{n12}(z)$ of the matrix $Y_n(z)$. From \eqref{RHP_1} we obtain that
for $z\in \Ga_0\cup\Ga_1$,
\begin{equation}\label{aop3}
Y_{n12+}(z)-Y_{n12-}(z)=\al(z)e^{-NV(z)}Y_{n11}(z),
\end{equation}
and from \eqref{RHP_3}, that as $z\to\infty$,
\begin{equation}\label{aop4}
Y_{n12}(z)=\mathcal O(z^{-n-1}).
\end{equation}
This implies that
\begin{equation}\label{aop5}
Y_{n12}(z)=\frac{1}{2\pi i}\int_\Ga \frac{e^{-NV(s)}Y_{n11}(s)}{s-z}\,ds,
\end{equation}
where again $\Ga=\Ga_0\cup\Ga_1$. By expanding $1/(s-z)$ into a geometric series, we obtain that
\begin{equation}\label{aop6}
Y_{n12}(z)=-\frac{1}{2\pi i}\sum_{k=0}^{n-1}\int_\Ga \frac{s^k e^{-NV(s)}Y_{n11}(s)}{z^{k+1}}\,ds
+\mathcal O(z^{-n-1}).
\end{equation}
Comparing this with \eqref{aop4}, we obtain that
\begin{equation}\label{aop7}
\int_\Ga s^k e^{-NV(s)}Y_{n11}(s)\,ds=0,\qquad k=0,1,\ldots, n-1,
\end{equation}
hence $Y_{n11}(z)$ is a monic orthogonal polynomial of degree $n$, $Y_{n11}(z)=P_n(z)$. Since the solution $Y_n(z)$ to
the RHP \eqref{RHP_1}-\eqref{RHP_3} is unique, the orthogonal polynomial $P_n(z)$ is unique as well.
\qed

{\it Proof of Proposition \ref{ttr}.} Consider the matrix-valued function
\begin{equation}\label{aop8}
U_n(z)=Y_{n+1}(z)Y_n(z)^{-1}.
\end{equation}
It has no jump, hence it is an entire function. As $z\to\infty$, $U_n(z)=\mathcal O(z)$, hence
$U_n(z)$ is a linear function. Moreover, \eqref{RHP_3} implies that $U_n(z)$ has the form
\begin{equation}\label{aop9}
U_n(z)=
\begin{pmatrix}
z+c_{11} & c_{12} \\
c_{21} & c_{22}
\end{pmatrix}.
\end{equation}
Considering the equation $Y_{n+1}(z)=U_n(z)Y_n(z)$ for the (1,1)-entry, we obtain the three term
recurrence relation \eqref{rr1} with some coefficients, which we denote $\be_n$ and $\ga_n^2$.
\qed

{\it Proof of Proposition \ref{str_1}.} Integrating by parts, we obtain that
\begin{equation}\label{aop10}
\begin{aligned}
0&=\int_\Ga P_n(z) P_n'(z) e^{-NV(z)}dz
=N\int_\Ga  P_n(z)^2 V'(z) e^{-NV(z)}dz\\
&=N\int_\Ga P_n(z)^2 (z-3uz^2) e^{-NV(z)}dz.
\end{aligned}
\end{equation}
Now, from \eqref{rr1} we have that
\begin{equation}\label{aop11}
\begin{aligned}
\int_\Ga  P_n(z)^2 z e^{-NV(z)}dz=h_n\be_n,
\end{aligned}
\end{equation}
and
\begin{equation}\label{aop12}
\begin{aligned}
\int_\Ga P_n(z)^2 z^2 e^{-NV(z)}dz
&=\int_\Ga [P_{n+1}(z)+\beta_nP_n(z)+\gamma_n^2 P_{n-1}(z)]^2  e^{-NV(z)}dz\\
&=h_{n+1}+\be_n^2h_n+\ga_n^4 h_n,
\end{aligned}
\end{equation}
hence
\begin{equation}\label{aop13}
\begin{aligned}
0=h_n\be_n-3u(h_{n+1}+\be_n^2h_n+\ga_n^4 h_{n-1}).
\end{aligned}
\end{equation}
This implies the first equation in \eqref{string}.

Similarly, integrating by parts, we obtain that
\begin{equation}\label{aop14}
\begin{aligned}
nh_{n-1}&=\int_\Ga P_{n-1}(z) P_n'(z) e^{-NV(z)}dz\\
&=N\int_\Ga P_{n-1}(z) P_n(z) V'(z) e^{-NV(z)}dz\\
&=N\int_\Ga P_{n-1}(z) P_n(z) (z-3uz^2) e^{-NV(z)}dz.
\end{aligned}
\end{equation}
Now, from \eqref{rr1} we have that
\begin{equation}\label{aop15}
\begin{aligned}
\int_\Ga P_{n-1}(z) P_n(z) z e^{-NV(z)}dz=h_n,
\end{aligned}
\end{equation}
and
\begin{equation}\label{aop16}
\begin{aligned}
&\int_\Ga P_{n-1}(z) P_n(z) z^2 e^{-NV(z)}dz\\
&=\int_\Ga [P_{n}(z)+\beta_{n-1}P_{n-1}(z)+\gamma_{n-1}^2 P_{n-2}(z)] P_n(z) z e^{-NV(z)}dz\\
&=(\be_n+\be_{n-1})h_n,
\end{aligned}
\end{equation}
hence
\begin{equation}\label{aop17}
\begin{aligned}
nh_{n-1}=N[h_n-3u(\be_n+\be_{n-1})h_n].
\end{aligned}
\end{equation}
This implies the second equation in \eqref{string}.
\qed

{\it Proof of Proposition \ref{an_1}.} Observe that the moments,
\begin{equation}\label{aan1}
c_j=\int_\Ga z^j e^{-N(z^2/2-uz^3)}dz,
\end{equation}
are $C^\infty$-functions of $u$ for $u\ge 0$, and they are analytic for $u>0$.
Hence, the same is true for the determinant $D_{n-1}$ and the polynomial $D_n(z)$, for any $n$.
\qed

\section{Proof of Toda equation \eqref{Toda_eq}}\label{A_Toda}

Let $N$ be fixed. Then, as $u\to 0$, the complex-valued measure
\[
e^{-N(\frac{z^2}{2}-uz^3)}dz
\]
on $\Ga=\al\Ga_0+(1-\al)\Ga_1$ converges to the Gaussian measure, and the moments
of the measure converge to the
moments of the Gaussian measure. This implies that there exists $u(N)>0$ such that for $u\in[0,u(N)]$
the orthogonal polynomials $P_n(z)$ exist  for $n=0,1,\ldots,N$. Hence
the orthogonal polynomials $\tilde P_n(\z)$ exist for
 $n=0,1,\ldots,N$, if $t\ge t(N)$, where
\begin{equation}\label{atoda1}
t(N)=\frac{1}{4[3 u(N)]^{4/3}}\,,
\end{equation}
cf. \eqref{FNFN2}
This implies, by the usual argument (see e.g. \cite{BI}), that Toda equation \eqref{Toda} is valid
for $t\ge t(N)$. Since both the free energy $\tilde F_N$ and the recurrence coefficient $\ga_N^2$ are analytic
in $t$ for $t>t_c$, we obtain, by the analytic continuation, that Toda equation \eqref{Toda} is valid
for $t>t_c$.   \qed

\section{Proof of formula \eqref{Th3_2}} \label{App_Th3_2}

In order to prove \eqref{Th3_2}, we will use formula \eqref{F32F21} and some linear transformations
and integral representations of the $_2F_1$ functions. Namely, we note that we can use \cite[15.8.7]{dlmf}:
$$
\mathop{_2F_1\/}\nolimits\!\left({-m,b\atop c};z\right)=\frac{(c-b)_{m}}{(c)_{m}}\mathop{_2F_1\/}\nolimits\!\left({-m,b\atop b-c-m+1};1-z\right),
$$
setting $m=j-1$ and $m=j-2$:
\begin{equation}
\begin{aligned}
_2F_1\left(\begin{array}{l} -j+1,2\\ -\frac{3j}{2}+1 \end{array};\frac{3}{2}\right)
&=\frac{\left(\frac{j}{2}+3\right)_{j-1}}{\left(\frac{j}{2}+1\right)_{j-1}}\,
_2F_1\left(\begin{array}{l} -j+1,2\\ \frac{j}{2}+3 \end{array};-\frac{1}{2}\right)\\
\end{aligned}
\end{equation}
and
\begin{equation}
\begin{aligned}
_2F_1\left(\begin{array}{l} -j+2,3\\ -\frac{3j}{2}+2 \end{array};\frac{3}{2}\right)
& =\frac{\left(\frac{j}{2}+4\right)_{j-2}}{\left(\frac{j}{2}+1\right)_{j-2}}\,
_2F_1\left(\begin{array}{l} -j+2,3\\ \frac{j}{2}+4 \end{array};-\frac{1}{2}\right).
\end{aligned}
\end{equation}
Additionally, we can use
\begin{equation}
\mathop{_2F_1\/}\nolimits\!\left({a,b\atop c};z\right)= \frac{\mathop{\Gamma\/}\nolimits\!\left(c\right)}{\mathop{\Gamma\/}\nolimits\!\left(b\right)\mathop{\Gamma\/}\nolimits\!\left(c-b\right)}\int _{0}^{1}\frac{t^{{b-1}}(1-t)^{{c-b-1}}}{(1-zt)^{a}}dt,
\end{equation}
which is valid provided that $\textrm{Re}\, c>\textrm{Re}\, b>0$. In our case,
\begin{equation}
 \begin{aligned}
_2F_1\left(\begin{array}{l} -j+1,2\\ \frac{j}{2}+3 \end{array};\frac{3}{2}\right)
&=\frac{\Gamma\left(\frac{j}{2}+3\right)}{\Gamma\left(\frac{j}{2}+1\right)}
\int _{0}^{1} t(1-t)^{\frac{j}{2}}\left(1+\frac{1}{2}t\right)^{j-1} dt,\\
_2F_1\left(\begin{array}{l} -j+2,3\\ \frac{j}{2}+4 \end{array};\frac{3}{2}\right)
&=\frac{\Gamma\left(\frac{j}{2}+4\right)}{2\,\Gamma\left(\frac{j}{2}+1\right)}
 \int _{0}^{1} t^2(1-t)^{\frac{j}{2}}\left(1+\frac{1}{2}t\right)^{j-2} dt.
\end{aligned}
\end{equation}
If we write the previous integrands as follows,
\begin{equation}\label{integ}
\begin{aligned}
I_1=\int _{0}^{1} t\left(1+\frac{1}{2}t\right)^{-1}e^{-j\phi(t)} dt, \qquad
I_2=\int _{0}^{1} t^2\left(1+\frac{1}{2}t\right)^{-2}e^{-j\phi(t)} dt,
\end{aligned}
\end{equation}
with
\begin{equation}
\phi(t)=-\frac{1}{2}\ln (1-t)-\ln\left(1+\frac{1}{2}t\right).
\end{equation}
This function has a minimum at $t=0$, since
\begin{equation}
\phi'(t)=\frac{3t}{2(1-t)(2+t)},
\end{equation}
which is positive for $0<t<1$. Therefore, for large $j$ the main contribution to both integrals comes from $t=0$. We make the change of variable
\begin{equation}
-\frac{1}{2}\ln (1-t)-\ln\left(1+\frac{1}{2}t\right)=\tau,
\end{equation}
which maps $t=0$ to $\tau=0$. We expand around $t=0$:
\begin{equation}
-\frac{1}{2}\ln (1-t)-\ln\left(1+\frac{1}{2}t\right)
=\frac{3}{8}t^2+\frac{1}{8}t^3+\frac{9}{64}t^4+\mathcal{O}(t^5),
\end{equation}
so inverting this series, we obtain
\begin{equation}
t=\frac{2\sqrt{6}}{3}\tau^{1/2}-\frac{4}{9}\tau-\frac{17\sqrt{6}}{81}\tau^{3/2}
+\mathcal{O}(\tau^2).
\end{equation}
Additionally, one can work out the following:
\begin{equation}
\begin{aligned}
t\left(1+\frac{t}{2}\right)^{-1}\frac{dt}{d\tau}&=\frac{4}{3}-\frac{8\sqrt{6}}{9}\tau^{1/2}+\mathcal{O}(\tau),\\
t^2\left(1+\frac{t}{2}\right)^{-2}\frac{dt}{d\tau}&=\frac{8\sqrt{6}}{9}\tau^{1/2}
-\frac{160}{27}\tau+\mathcal{O}(\tau^{3/2}).
\end{aligned}
\end{equation}
An application of classical Watson's lemma \cite{Olver}, shows now that the integrals in \eqref{integ} are $I_1=\mathcal{O}(j^{-1})$ and $I_2=\mathcal{O}(j^{-3/2})$ respectively. The factors that multiply in front are
\begin{equation}
 \frac{\left(\frac{j}{2}+3\right)_{j-1}}{\left(\frac{j}{2}+1\right)_{j-1}}
\frac{\Gamma\left(\frac{j}{2}+3\right)}{\Gamma\left(\frac{j}{2}+1\right)}
=\frac{\Gamma\left(\frac{3j}{2}+2\right)}{\Gamma\left(\frac{3j}{2}\right)}
\sim \frac{9j^2}{4}
\end{equation}
and
\begin{equation}
 \frac{\left(\frac{j}{2}+4\right)_{j-2}}{\left(\frac{j}{2}+1\right)_{j-2}}
\frac{\Gamma\left(\frac{j}{2}+4\right)}{2\Gamma\left(\frac{j}{2}+1\right)}
=\frac{\Gamma\left(\frac{3j}{2}+2\right)}
{2\Gamma\left(\frac{3j}{2}-1\right)}
\sim \frac{27j^3}{16},
\end{equation}
so we only need to consider the second integral $I_2$, which is dominant for large $j$. Writing everything together and noting that the binomial number has the following behavior,
\begin{equation}
 {\frac{3j}{2}-1 \choose j-1}=\frac{\Gamma\left(\frac{3j}{2}\right)}{\Gamma(j)\Gamma\left(\frac{j}{2}+1\right)}\sim \frac{2^{-j+\frac{1}{2}}3^{\frac{3j}{2}-\frac{1}{2}}j^{-1/2}}{\sqrt{\pi}}, \qquad j\to\infty,
\end{equation}
we obtain
\begin{equation}
\begin{aligned}
\frac{45\cdot 72^j}{4} \frac{\Gamma\left(\frac{3j}{2}\right)}{\Gamma(j)\Gamma\left(\frac{j}{2}+1\right)}
\frac{6(j-1)}{5(3j-2)}\, _2F_1\left(\begin{array}{l} -j+2,3\\ -\frac{3j}{2}+2 \end{array};\frac{3}{2}\right)
\sim 2^{2j-2} 3^{\frac{7j}{2}+3}j,
\end{aligned}
\end{equation}
when $j\to\infty$. We observe that we can write this as
\begin{equation}
2^{-2} 3^3 j (3^\frac{7}{2}2^2)^j=\frac{27j}{4w_c^j},
\end{equation}
and recalling \eqref{F2u} we obtain
\begin{equation}
f^{(2)}_{2j}\sim \frac{27j(2j)!}{4\cdot 36\cdot 9j^2 w_c^j}=\frac{(2j)!}{48j w_c^j},
\end{equation}
which proves asymptotic formula \eqref{Th3_2}.

\end{document}